\documentclass[runningheads]{llncs}

\usepackage{stmaryrd} 
\usepackage{amsmath}
\usepackage{amssymb}
\usepackage{xcolor}
\usepackage{xspace}
\usepackage{graphicx}
\usepackage{complexity}
\usepackage{bm}
\usepackage{tikz}
\usetikzlibrary{positioning, shapes, calc}
\usepackage{microtype}

\newtheorem{fact}[theorem]{Fact}

\newcommand{\defeq}{\stackrel{\scriptscriptstyle\text{def}}{=}}

\newcommand{\N}{\mathbb{N}}                    
          
\newcommand{\Mv}{M_{\vec{v}}}
       
\newcommand{\multiset}[1]{\Lbag#1\Rbag}        

\renewcommand{\PP}{\mathcal{P}}                  
\newcommand{\trans}[1]{\xrightarrow{#1}}       

\newcommand{\pre}{\mathit{pre}} 
\newcommand{\post}{\mathit{post}} 

\newcommand{\preset}[1]{{}^\bullet #1}
\newcommand{\postset}[1]{{#1}^\bullet}

\newcommand{\unorm}[1]{\|{#1}\|_u}
\newcommand{\lnorm}[1]{\|{#1}\|_l}
\newcommand{\sem}[1]{\llbracket{#1}\rrbracket}

\usepackage{thm-restate}

\usepackage{subcaption}
\begin{document}
\title{Parameterized Analysis of Immediate Observation Petri Nets \thanks{This project has received funding from the European Research Council (ERC) under the European Union's Horizon 2020 research and innovation programme under grant agreement No 787367 (PaVeS)}}
\titlerunning{Parameterized Analysis of Immediate Observation Petri Nets}
\author{Javier Esparza\inst{1}\orcidID{0000-0001-9862-4919} \and
Mikhail (Michael) Raskin\inst{2,3}\orcidID{0000-0002-6660-5673} \and
Chana Weil-Kennedy\inst{3}\orcidID{0000-0002-1351-8824}}
\authorrunning{J. Esparza et al.}
\institute{Technical University of Munich
\email{esparza@in.tum.de} \and 
Technical University of Munich
\email{raskin@in.tum.de} \and 
Technical University of Munich
\email{chana.weilkennedy@in.tum.de} 
}
\maketitle 
\begin{abstract}
We introduce immediate observation Petri nets, a class of interest in the study of population protocols (a model of distributed computation), and enzymatic chemical networks. In these areas, relevant analysis questions translate into \emph{parameterized} Petri net problems: whether an infinite set of Petri nets with the same underlying net, but different initial markings, satisfy a given property.
We study the parameterized reachability, coverability, and liveness problems for immediate observation Petri nets. We show that all three problems are in \PSPACE \ for infinite sets of initial markings defined by counting constraints, a class sufficiently rich for the intended application. This is remarkable, since the problems are already \PSPACE-hard when the set of markings is a singleton, i.e., in the non-parameterized case. We use these results to prove that the correctness problem for immediate observation population protocols is \PSPACE-complete, answering a question left open in a previous paper. 

\end{abstract}
\section{Introduction}
We study the theory of \emph{immediate observation Petri nets}, a class of Petri nets with 
applications to the study of population protocols and chemical reaction networks, two models of distributed computation.

Population protocols are a formalism for the study of ad hoc networks of tiny computing devices without any infrastructure. They were introduced by Angluin \textit{et al.}~\cite{AADFP04}, and have been very intensely studied, in particular in recent years (see e.g.~\cite{AlistarhAEGR17,AlistarhAG18,AlistarhG18,ElsasserR18}). The
model postulates a ``soup'' of finite-state, indistinguishable agents interacting in pairs.  Formally, a population protocol has a finite set of states $Q$ and a set of transitions of the form \((q_1, q_2) \mapsto (q_3, q_4)\), which allow two agents in states \(q_1\) and \(q_2\) to interact and simultaneously move to \(q_3\) and \(q_4\).
A global state of the protocol, called a \emph{configuration}, is a mapping $C$ that assigns to each state $q$ the current number $C(q)$ of agents in $q$.  
A protocol has a set of initial configurations. Intuitively, each initial configuration corresponds to an input, and the purpose of a protocol is to compute a boolean output, $0$ or $1$, for each input. A protocol outputs $b$ for a given initial configuration $C$ if in all fair runs starting at $C$ (with respect to a certain fairness condition), all agents eventually agree to output $b$. So, loosely speaking, population protocols compute by reaching a stable consensus. The \emph{predicate computed by a protocol} is the function that assigns to each initial configuration $C$ the boolean output computed by the protocol when started at $C$.

Even this very abstract description shows that a population protocol is ``nothing but'' a (place/transition) Petri net: a state corresponds to a place, a transition of the protocol to a  net transition with two input and two output places, an agent to a token, and a configuration to a marking.
In the last years, this connection was exploited to address the problem of proving population protocols correct.
The fundamental correctness problem for population protocols asks, given a protocol and a predicate, whether the protocol computes the predicate. 
This question  was proved decidable in \cite{EsparzaGLM15,EsparzaGLM17}, but, unfortunately, the same papers also showed that the correctness problem is at least as hard as Petri net reachability, and so of non-elementary complexity \cite{CLLLM18}.

In their seminal paper on the expressive power of population protocols \cite{AAER07}, Angluin \textit{et al.} defined subclasses corresponding to different communication primitives between agents. In the standard model, agents communicate through \emph{rendez-vous}: transitions \((q_1, q_2) \mapsto (q_3, q_4)\) formalize that both partners exchange full information about their current states, and update them based on it. Angluin \textit{et al.} introduced \emph{immediate observation} protocols, called IO protocols for short, whose transitions have the form $(q_1, q_2) \mapsto (q_1, q_3)$. 
Intuitively, in an IO protocol an agent can change its state from $q_2$ to $q_3$ by \emph{observing} that another agent is in state $q_1$; the agent in state $q_1$ may not even know that it is being observed. 
A characterization of the predicates computable by IO protocols was given in \cite{AAER07}, and in \cite{EsparzaGMW18} Esparza \textit{et al.} studied the complexity of the correctness problem.
They showed that it was \PSPACE-hard and solvable in \EXPSPACE, and left the problem of closing this gap for future research. 

In this paper we study the theory of \emph{immediate observation Petri nets} (IO nets), the Petri nets underlying immediate observation protocols. Our initial motivation is their application to population protocol problems, especially the gap just mentioned. However, IO nets also model \emph{networks of enzymatic chemical reactions}, in which an enzyme $E$ catalyzes the formation of product $P$ from substrate $S$ \cite{baldan2010petri,marwan2011petri}. An example of application of Petri net techniques to such a network is presented in \cite{angeli2007petri}.\footnote{The Petri nets of \cite{angeli2007petri} are in fact slightly more general than IO nets, but equivalent to them for properties that depend only on the reachability graph, as are the net properties studied in \cite{angeli2007petri}.}

Analysis problems for population protocols or chemical networks are  \emph{parametric} in the number of agents or the number of molecules. In other words, they ask whether the system satisfies a property \emph{for any number of agents} or \emph{for any number of molecules}. When formalized as Petri nets problems, they become questions of the form ``does an \emph{infinite set} of Petri nets differing only in their initial markings satisfy a given property?'' We investigate parameterized versions of the standard reachability, coverability, and liveness problems for IO nets in which the set of initial markings is a \emph{cube}, i.e., a set of markings obtained by attaching to each place a lower bound and an upper bound (possibly infinite) for the number of tokens. We prove that, remarkably, while the standard problems are \PSPACE-hard even in the non-parameterized case, they remain in \PSPACE \ in the parameterized case. This is in strong contrast with the situation for more general classes of nets.
For example, while  the non-parametric problems are in \PSPACE\  for conservative nets or 1-safe nets, their ``cube-versions'' become  \EXPSPACE-hard or even non-elementary.  As an application of our results, we close the gap left open in \cite{EsparzaGMW18}, and prove that the correctness problem for IO protocols is \PSPACE-complete. 

For space reasons, all missing proofs and some technical details are relegated to the full version of this article \cite{EsparzaRW19}.
\section{Preliminaries}
\paragraph{Multisets.} A \emph{multiset} on a finite set \(E\) is a mapping \(C \colon E \rightarrow \N\), i.e. for any $e\in E$, \(C(e)\) denotes the number of occurrences of element \(e\) in \(C\).
Let $\multiset{e_1,\ldots,e_n}$ denote the multiset $C$ such that $C(e)=|\{j\mid e_j=e\}|$.
Operations on \(\N\) like addition or comparison are extended to multisets by defining them component wise on each element of \(E\).
Subtraction is allowed as long as each component stays non-negative.
We define $|C|\defeq\sum_{e\in E} C(e)$ the sum of the occurrences of each element in $C$. 
Given a total order $e_1 \prec e_2 \prec \cdots \prec e_n$ on $E$, a multiset $C$ can be 
equivalently represented by the vector $(C(e_1), \ldots, C(e_n))\in \N^n$.

\paragraph{Place/transition Petri nets with weighted arcs.} A \emph{Petri net} $N$ is a triple $(P,T,F)$ consisting of a finite set of \emph{places} $P$, a finite set of \emph{transitions} $T$ and a \emph{flow function} $F \colon (P \times T) \cup (T \times P) \rightarrow \mathbb{N}$. 

A \emph{marking} $M$ is a multiset on $P$, and we say that a marking $M$ puts $M(p)$ \emph{tokens} in place $p$ of $P$. The \emph{size} of $M$, denoted by $|M|$, is the total number of tokens in $M$.
The \emph{preset} $\preset{t}$ and \emph{postset} $\postset{t}$ of a transition $t$ are the multisets on $P$ given by $\preset{t}(p)=F(p,t)$ and $\postset{t}(p)=F(t,p)$. A transition $t$ is \emph{enabled} at a marking $M$ if $\preset{t} \leq M$, i.e. $\preset{t}$ is component-wise smaller or equal to $M$.
If $t$ is enabled then it can be \emph{fired}, leading to a new marking $M'=M - \preset{t} + \postset{t}$. 
We note this $M \xrightarrow{t} M'$.

\paragraph{Reachability and coverability}
Given $\sigma=t_1 \ldots t_n$ we write $M \xrightarrow{\sigma} M_n$ when $M \xrightarrow{t_1} M_1 \xrightarrow{t_2} M_2 \ldots \xrightarrow{t_n} M_n$, and call $\sigma$ a \emph{firing sequence}. 
We write $M' \trans{*} M''$ if $M' \xrightarrow{\sigma} M''$ for some $\sigma \in T^*$, and say that $M''$ is \emph{reachable} from $M'$. 
A marking $M$ \emph{covers} another marking $M'$, written $M \geq M'$ if $M(p) \geq M'(p)$ for all places $p$.  
A marking $M$ is \emph{coverable} from $M'$ if there exists a marking $M''$ such that $M' \trans{*} M'' \geq M$.

\paragraph{Conservative Petri nets.}A \emph{Petri net} $N=(P,T,F)$ is \emph{conservative} if there is a mapping $I \colon P \rightarrow \mathbb{Q}_{>0}$ such that $\sum_{p \in P} I(p) \cdot \preset{t}(p) = \sum_{p \in P} I(p) \cdot \postset{t}(p)$ for all $t$. Further, $N$ is \emph{$\vec{1}$-conservative} if it is conservative with $I$ equal to $1$ over all $P$ (see \cite{MayrW14}). It follows immediately from the definitions that if $N$ is conservative and $M \trans{*} M'$, then 
$\sum_{p \in P} I(p) \cdot M(p) = \sum_{p \in P} I(p) \cdot M'(p)$.
\section{A Primer on Population Protocols}
\label{IOPPprimer}

As mentioned in the introduction, a  population protocol consists of a set of states $Q$ and a set of transitions 
$T \subseteq Q^2 \times Q^2$. A transition $\big( (q_1, q_2), (q_3,q_4) \big) \in T$ is denoted $(q_1, q_2)\mapsto (q_3, q_4)$.
A \emph{configuration} is a multiset of states. A configuration, say $C$, such that $C(q_1)=2$ and $C(q_2)=1$,
indicates that currently there are two agents in state $q_1$ and one agent in state $q_2$. 
The connection to Petri nets is immediate: The Petri net modeling a protocol has one place for each state, and one transition for every transition of the protocol. If transition $t$ of the Petri net models $(q_1, q_2)\mapsto (q_3, q_4)$, then $\preset{t}=\multiset{q_1, q_2}$, and $\postset{t}=\multiset{q_3,q_4}$.
An agent in state $q$ is modeled by a token in place $q$. A configuration $C$ with $C(q)$ agents in state $q$ is modeled by the marking putting $C(q)$ tokens in place $q$ for every $q \in Q$. 
Observe that the transitions of the net do not change the total number of tokens, and so we have:

\begin{fact}\label{fact:cons}
Petri nets obtained from population protocols are $\vec{1}$-conservative.
\end{fact}

Population protocols are designed to compute predicates $\varphi \colon \mathbb{N}^k \rightarrow \{0,1\}$. We first give an informal explanation of how a protocol computes a predicate, and then a formal definition using Petri net terminology. A protocol for $\varphi$ has a distinguished set of input states $\{q_1, q_2, \ldots, q_k\} \subseteq Q$. Further, each state of $Q$, initial or not, is labeled with an \emph{output}, either $0$ or $1$. Assume for example $k=2$. In order to compute $\varphi(n_1, n_2)$, we first place $n_i$ agents in $q_i$  for $j=1,2$, and $0$ agents in all other states. This is the initial configuration of the protocol for the input $(n_1, n_2)$. Then we let the protocol run. The protocol satisfies that in every \emph{fair} run starting at the initial configuration (fair runs are defined formally below), eventually all agents reach states labeled with $1$, and stay in such states forever, or they reach states of labeled with $0$, and stay in such states forever. So, intuitively, in all fair runs all agents eventually ``agree'' on a boolean value. By definition, this value is the result of the computation, i.e, the value of $\varphi(n_1, n_2)$.

Formally, and in Petri net terms, fix a Petri net $N=(P, T, F)$ with $|\preset{t}|=2=|\postset{t}|$ for every transition $t$.
Further, fix a set $I=\{p_1, \ldots , p_k\}$ of \emph{input places}, and a \emph{function} $O \colon P \rightarrow \{0,1\}$.  A marking $M$ of $N$ is a \emph{$b$-consensus} if $M(p)>0$ implies $O(p)=b$. A $b$-consensus $M$ is \emph{stable} if every marking reachable from $M$ is also a $b$-consensus. A firing sequence $M_0 \trans{t_1} M_1 \trans{t_2} M_2 \cdots$ of $N$ is \emph{fair} if it is finite and ends at a deadlock marking, or if it is infinite and the following condition holds for all markings $M,M'$ and $t \in T$: if  $M \trans{t} M'$ and $M=M_i$ for infinitely many $j \geq 0$, then $M_j \trans{t_{j+1}}M_{j+1} \, = \, M \trans{t} M'$ for infinitely many $j \geq 0$. 
In other words, if a fair sequence reaches a marking infinitely often, then all the transitions enabled at that marking will be fired infinitely often from that marking.
A fair firing sequence \emph{converges to $b$} if there is $j \geq 0$ such that $M_{j}$ is a $b$-consensus for every marking $j \geq i$ of the sequence. For every $\vec{v} \in \mathbb{N}^k$ with $|\vec{v}| \geq 2$ let $\Mv$ be the marking given by $\Mv(p_i) = \vec{v}_i$ for every $p_i \in I$, and $\Mv(p) = 0$ for every $p \in  P \setminus I$. We call $\Mv$ the \emph{initial marking for input $\vec{v}$}. The net \emph{$N$ computes the predicate $\varphi \colon \mathbb{N}^k \rightarrow \{0,1\}$} if for every $\vec{v} \in \mathbb{N}^k$, every fair firing sequence starting at $\Mv$ converges to $b$.

\begin{figure}
\vskip-0.5cm
\centering%
\begin{subfigure}[t]{0.49\textwidth}
\centering%
\resizebox{3cm}{!}{
\begin{tikzpicture}

\node[circle, draw]
(Node1)
{$q_1$};
\node[below=0.8 of Node1, rectangle, draw]
(Trans1)
{$t_1$};
\node[left=1.0 of Trans1, rectangle, draw]
(Trans3)
{$t_3$};
\node[below=0.4 of Trans1, circle, draw]  
(Node2)
{$q_2$};
\node[below=0.8 of Node2, rectangle, draw]
(Trans2)
{$t_2$};
\node[right=1.0 of Trans2, rectangle, draw]
(Trans4)
{$t_4$};
\node[below=0.4 of Trans2, circle, draw]
(Node3)
{$q_3$};  

\draw[->] (Node1) to node[left] {$2$} (Trans1);
\draw[->,bend right=20] (Trans1) to (Node1);
\draw[->] (Trans1) to (Node2);
\draw[->] (Node2) to node[left] {$2$} (Trans2);
\draw[->,bend right=20] (Trans2) to (Node2);
\draw[->] (Trans2) to (Node3);
\draw[->] (Node1) to (Trans3);
\draw[->] (Node2) to (Trans4);
\draw[->,bend right=20] (Node3) to (Trans4);
\draw[->] (Trans4) to node[above] {$2$} (Node3);
\draw[->,bend left=20] (Node3) to (Trans3);
\draw[->] (Trans3) to node[left] {$2$} (Node3);

\end{tikzpicture}
}%
\caption{Net for $\PP_1$}\label{figure-running-example}%
\end{subfigure}\hfill%
\begin{subfigure}[t]{0.49\textwidth}
\centering%
\resizebox{5cm}{!}{
\begin{tikzpicture}
\node[circle, draw]%
(Node0)%
{$q_0$};%
\node[above left=1.5 and 2 of Node0, circle, draw]%
(Node1)%
{$q_1$};%
\node[above right=1.5 and 2 of Node0, circle, draw]%
(Node2)%
{$q_2$};%
\node[below=0.4 of Node0, rectangle, draw]%
(Trans5)%
{$t_5$};%
\node[below=0.8 of Trans5, circle, draw]%
(Node3)%
{$q_3$};%
\node[left=2 of Trans5, rectangle, draw]%
(Trans4)%
{$t_4$};%
\node[right=2 of Trans5, rectangle, draw]%
(Trans6)%
{$t_6$};%
\node[above=1.4 of Node0, rectangle, draw]%
(Trans2)%
{$t_2$};%
\node[below left=1.0 and 1.2 of Trans2, rectangle, draw]%
(Trans1)%
{$t_1$};%
\node[below right=1.0 and 1.2 of Trans2, rectangle, draw]%
(Trans3)%
{$t_3$};%
\draw[->] (Node1) to (Trans4);%
\draw[->] (Node0) to (Trans5);%
\draw[->] (Node2) to (Trans6);%
\draw[->, bend left=20] (Node3) to (Trans4);%
\draw[->, bend right=10] (Node3) to (Trans5);%
\draw[->, bend right=20] (Node3) to (Trans6);%
\draw[->] (Trans4) to node[above] {$2$} (Node3);%
\draw[->] (Trans5) to node[left] {$2$} (Node3);%
\draw[->] (Trans6) to node[above] {$2$} (Node3);%
\draw[->] (Trans2) to (Node2);%
\draw[->] (Trans2) to (Node0);%
\draw[->] (Node1) to node[above] {$2$} (Trans2);%
\draw[->] (Node2) to node[left] {$2$} (Trans3);%
\draw[->] (Trans3) to (Node0);%
\draw[->] (Trans3) to (Node3);%
\draw[->] (Node1) to (Trans1);%
\draw[->] (Node2) to (Trans1);%
\draw[->] (Trans1) to (Node0);%
\draw[->] (Trans1) to (Node3);%
\end{tikzpicture}
}%
\caption{Net for $\PP_2$}\label{figure-second-protocol}%
\end{subfigure}%
\caption{Petri nets underlying population protocols.}%
\vskip-0.5cm
\end{figure}
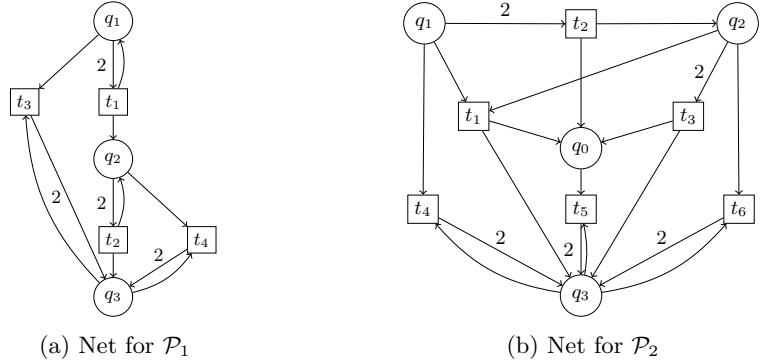

\begin{example}
\label{ex:twoPP}
We exhibit two population protocols that compute the predicate $\varphi(x) \defeq [x \geq 3]$, and their corresponding Petri nets.

The first protocol $\PP_1$ has states $Q_1= \{ q_1,q_2,q_3 \}$ and transitions 
$(q_a, q_a)\mapsto(q_{a+1},q_a)$ and $(q_a,q_3)\mapsto(q_3,q_3)$ for $a=1,2$.
The only input state is $q_1$. 
States $q_1$ and $q_2$ are labeled with $0$, and state $q_3$ with $1$.
The Petri net for $\PP_1$ is shown in Figure \ref{figure-running-example}.
The initial marking for input $x$ puts $x$ tokens on $q_1$, and no token elsewhere.
If $x \geq 3$, then every fair firing sequence eventually reaches the deadlock marking with $x$ tokens in $q_3$ and no tokens elsewhere (indeed, transitions $t_3$ and $t_4$ ensure that after a token reaches $q_3$, eventually all other tokens move to $q_3$ as well). So the agents eventually reach consensus $1$. If $x < 3$, then no firing sequence ever puts a token in $q_3$ and so, since both $q_1$ and $q_2$ have output $0$, the agents reach consensus $0$.

The second protocol $\PP_2$ has place set $Q_2= \{ q_0,q_1,q_2,q_3 \}$, and transitions $(q_a, q_b)\mapsto(q_0,q_{\min(a+b,3)})$ for $0 < a,b < 3$, 
and ($q_a, q_3)\mapsto(q_3,q_3)$ for $0\leq a < 3$. 
The Petri net for $\PP_2$ is shown in Figure \ref{figure-second-protocol}.
Again, the only input state is $q_1$. States $q_0, q_1, q_2$ are labeled with $0$, and state $q_3$ is labeled with $1$. 
The reader can check that, as in the first protocol, the agents eventually reach consensus $1$ from an input $x$ if{}f $x \geq 3$.

Both these protocols could be generalized to calculate $[x \geq n]$ for any natural $n \geq 1$.
\end{example}

\subsubsection{Immediate observation protocols}

When two agents of a population protocol communicate, they can both simultaneously change their states. This corresponds to communication by \emph{rendez-vous}. In \cite{AAER07}, Angluin \textit{et al.} introduced \emph{immediate observation} protocols, corresponding to a more restricted communication mechanisms. One of the agents observes the state of the other agent, and updates its own state accordingly; the observed agent does not change its state, since it may not even know that it is being observed. Transitions are of the form $(q_s, q_o) \mapsto (q_d, q_o)$, where $q_o$ is the state of the observed agent. In the paper, they showed that the predicates computable by immediate observation protocols are exactly those described by counting constraints, a formalism introduced in Section \ref{CC}.

\begin{example}
Protocol $\PP_1$ of Example \ref{ex:twoPP} is immediate observation, but $\PP_2$ is not.
\end{example}

\subsubsection{Verifying population protocols} 
Not every population protocol is well designed. For some inputs $(n_1, \ldots, n_k)$ the protocol can have fair runs that never converge, or fair runs converging to the wrong value 
$1 -\varphi(n_1, \ldots, n_k)$. This raises the question of how to automatically verify that a protocol correctly computes a predicate. The main difficulty is to prove convergence to the right value \emph{for each of the infinitely many possible inputs}. In Petri net terms, we have to show that the net derived from the protocol satisfies a property \emph{for infinitely many initial markings}. So, strictly speaking, we have to show that an infinite collection of Petri nets satisfies a given property. We call problems of this kind \emph{parameterized}.

\section{Parameterized Analysis Problems}
\label{Parameterized}
Standard analysis problems for Petri nets concern one initial marking. For example,
the \emph{reachability problem} (\emph{coverability problem}) consists of, given a net $N$ and two markings $M, M'$ of $N$, deciding if $M$ is reachable (coverable) from $M'$.Parameterized problems, like the correctness problem for population protocols, involve an \emph{infinite} set of initial markings. In order to study their complexity, it is necessary to specify the shape of the set. For the applications to population protocols and chemical networks the following definition is adequate:

\begin{restatable}{definition}{DefCube}
\label{def:cube}
A set ${\cal M}$ of markings of a net $N$ is a \emph{cube} if there 
are mappings $L \colon P \rightarrow \mathbb{N}$ and $U \colon P \rightarrow \mathbb{N} \cup \{\infty\}$ such that $M \in {\cal M}$ if{}f $L \leq M \leq U$. We call $L$ and $U$ the \emph{lower} and \emph{upper bound} of ${\cal M}$, respectively, and use the notation $(L,U) \defeq {\cal M}$.
The cube-reachability (cube-coverability) problem consists of, given a net $N$ and cubes ${\cal M}, {\cal M}'$ of $N$,  deciding if there are markings $M \in {\cal M}, M' \in {\cal M}'$, such that $M$ is reachable (coverable) from $M'$. 
\end{restatable}

Observe that, if the set of places of the Petri net corresponding to a population protocol is $\{p_1, \ldots, p_n, p_{n+1}, \ldots, p_{n+m}\}$, where $p_1, \ldots, p_n$ are the initial places, then the set of input configurations corresponds to the cube $(L, U)$ where $L(p_i)=0$ for every $1 \leq i \leq n+m$, and $U(p_i) = \infty$ for $1 \leq i \leq n$ and $U(p_i) = 0$ for $n+1 \leq i \leq n+m$.

In general, parameterized problems are much harder than non-parameterized ones. 
Consider for example the class of conservative Petri nets which, by Fact \ref{fact:cons} contains all nets derived from population protocols. We have:

\begin{restatable}{theorem}{ThmConservativeNets}
For $\vec{1}$-conservative Petri nets:
\begin{itemize}
\item Reachability, coverability, and liveness are in \PSPACE.
\item Cube-reachability and cube-coverability are as hard as for general Petri nets, and so non-elementary and \EXPSPACE-hard, respectively. 
\end{itemize}
\end{restatable}

In the rest of the paper we introduce immediate observation Petri nets, the class of Petri nets corresponding to immediate observation protocols and enzymatic reaction networks, and study the cube-reachability, coverability, and liveness problems. We prove that, while the problems are \PSPACE-hard even for single markings, their cube versions \emph{remain} \PSPACE. This pinpoints the essential property of the class: loosely speaking, deciding standard problems for infinitely many markings is not harder than deciding them for one marking.

\section{Immediate Observation Petri Nets}
We introduce the class of immediate observation Petri nets (IO nets) and then show that the reachability, coverability, and liveness problems are \PSPACE-hard for this class.

\begin{definition}
\label{def:IOnet}
A transition $t$ of a Petri net is an \emph{immediate observation transition} if there are three 
places $p_s,p_d, p_o$, not necessarily distinct, such that 
$\preset{t}=\multiset{p_s,p_o}$ and $\postset{t}=\multiset{p_d,p_o}$.
We call $p_s, p_d, p_o$ the \emph{source}, \emph{destination}, and \emph{observed} places of $t$, respectively. A Petri net is an \emph{immediate observation net} if and only if all its transitions are immediate observation transitions.
\end{definition}

Following the useful convention of population protocols, we  write  $t = (p_s,p_o)\mapsto (p_d,p_o)$.


\begin{example}
The Petri net illustrated in Figure \ref{figure-running-example} is an immediate observation Petri net.
\end{example}


\label{IO-nets-hardness}
We show that the standard simulation of bounded-tape Turing machines 
by  1-safe Petri nets, as described for example in \cite{ChengEP95,Esparza96},
can be modified to produce an IO net (actually, a 1-safe IO net). Using this result, we 
can then easily prove that the reachability, coverability, and liveness
problems are \PSPACE-hard. Since a set consisting of a single marking is a special case of a cube,
the result carries over to the cube-versions of the problems.

We fix a deterministic Turing machine $M$ with set of control states $Q$, alphabet $\Sigma$ containing the empty symbol $\text{\textvisiblespace}$, and partial transition function $\delta \colon  Q\times\Sigma \to Q\times\Sigma\times D$ ($D=\{-1,+1\}$). We let $K$ denote an upper bound on the number of tape cells visited by the computation of $M$ on empty tape. The \emph{implementation} of $M$ is the IO Petri net $N_M$ described below.
        
        \newcommand{\pass}[2]{\textit{off}[#1,#2]}
        \newcommand{\act}[2]{\textit{on}[#1,#2]}
        \newcommand{\stable}[2]{\textit{at}[#1,#2]}
        \newcommand{\switch}[4]{\textit{move}[#1,#2,#3,#4]}

        \noindent \textbf{Places of $N_M$.} The net  $N_M$ contains two sets of \emph{cell places}
        and \textit{head places} modelling the state of the tape cells and the head, respectively. The cell places are:
        \begin{itemize}
        \item $\pass{\sigma}{n}$ for each $\sigma\in \Sigma$ and $1\leq n\leq K$. 
        A token on $\pass{\sigma}{n}$ denotes that cell $n$ contains symbol $\sigma$, and the cell is ``off'', i.e., the head is not on it.
        \item $\act{\sigma}{n}$ for each $\sigma\in \Sigma$ and $1\leq n\leq K$, with analogous intended meaning.
        \end{itemize}
        The head places are:
        \begin{itemize}
        \item $\stable{q}{n}$ for each $q\in Q$ and $1\leq n\leq K$.  A token on $\stable{q}{n}$ denotes that the head is in control state $q$ and at cell $n$.
        \item $\switch{q}{\sigma}{n}{d}$ for each $q\in Q$, $\sigma\in\Sigma$, 
        $1\leq n\leq K$
        and every $d\in D$ such that $1\leq n+d\leq K$. 
        A token on $\switch{q}{\sigma}{n}{d}$ denotes that head is in control state $q$, has left cell $n$ after writing symbol $\sigma$ on it, and is currently moving in the direction given by $d$. 
        \end{itemize}

       \noindent \textbf{Transitions of $N_M$.} Intuitively, the implementation of $M$ contains a  set of \emph{cell transitions} in which a cell observes the head and changes its state, and a set of \emph{head transitions} in which the head observes a cell. Further, each of these sets contains transitions of two types. The set of cell transitions contains:
        \begin{itemize}
        \item Type  \textbf{1a}: $(\pass{\sigma}{n} \ , \ \stable{q}{n}) \mapsto (\act{\sigma}{n} \ , \ \stable{q}{n})$ for every state $q \in Q$, symbol $\sigma \in \Sigma$, and cell $1 \leq n \leq K$. \\
        The $n$-th cell, currently \textit{off}, observes that the head is on it, and switches itself \textit{on}.
        \item Type \textbf{1b}: $(\act{\sigma}{n} \ , \  \switch{q}{\sigma'}{n}{d}) \mapsto ( \pass{\sigma'}{n} \ , \ \switch{q}{\sigma'}{n}{d})$ for every $q \in Q$, $\sigma \in \Sigma$, and $1 \leq n \leq K$ such that $1 \leq n+d \leq K$. \\
        The $n$-th cell, currently \textit{on}, observes that the head has left after writing $\sigma'$, and switches itself \textit{off} (accepting the character the head intended to write).
        \end{itemize}
        The set of head transitions contains:
        \begin{itemize}
        \item Type \textbf{2a}: $(\stable{q}{n}  \ , \ \act{\sigma}{n}) \mapsto (\switch{\delta_Q(q,\sigma)}{\delta_\Sigma(q,\sigma)}{n}{\delta_D(q,\sigma)} \ , \ \act{\sigma}{n})$ for every $q \in Q$, $\sigma \in \Sigma$, and $1 \leq n \leq K$ such that $1\leq n+\delta_D(q,\sigma)\leq K$. \\
        The head, currently on cell $n$, observes that the cell is \textit{on}, writes the new symbol on it, and leaves.
        \item Type \textbf{2b}: $(\switch{q}{\sigma}{n}{d}  \ , \ \pass{\sigma}{n}) \mapsto (\stable{q}{n+d} \ , \ \pass{\sigma}{n})$ for every $q \in Q$, $\sigma \in \Sigma$, and $1 \leq n \leq K$ such that $1 \leq n+d \leq K$.\\
        The head, currently moving, observes that the old cell has turned \textit{off}, and places itself on the new cell.
        \end{itemize}

This concludes the definition of $N_M$. 
In Theorem \ref{thm:simul} below we formalize the relation between the Turing machine $M$ and its implementation $N_M$, using the following definition. 
        
\begin{definition}
Given a configuration $c$ of $M$ with control state $q$, tape content $\sigma_1\sigma_2\cdots \sigma_K$, and head on cell $n \leq K$, we denote $M_c$ the marking that puts a token in $\pass{\sigma_i}{i}$ for each $1\leq i\leq K$, a token in $\stable{q}{n}$, and no tokens elsewhere.
\end{definition}

Now we state our simulation theorem and hardness result.

\begin{restatable}{theorem}{theoremSimulationStep}
\label{thm:simul}
For every two configurations $c, c'$ of $M$ that write at most $K$ cells:  $c \trans{} c'$ if{}f $M_c \trans{t_1t_2t_3t_4} M_{c'}$ in $N_M$ for some transitions $t_1, t_2, t_3, t_4$ of types \textbf{1a}, \textbf{2a}, \textbf{1b}, \textbf{2b}, respectively.
\end{restatable}

\begin{restatable}{theorem}{ThmReachabilityHard}\label{thm:liveness-hard}
The reachability, coverability and liveness problems for  IO nets are \PSPACE-hard.
\end{restatable}

\section{The Pruning Theorem}
\label{Pruning}
In this section, we present the fundamental property of immediate observation nets that entails most of the results in this paper: the Pruning Theorem. 

The Pruning Theorem intuitively states that if $M$ is coverable from a marking $M''$, then it is also coverable from a ``small" marking $S'' \leq M''$, where ``small" means $|S''| \leq |M| + |P|^3$. 
We state the theorem below, and then build up to its proof which is presented in Section \ref{subsec-proof}.

\begin{theorem} [Pruning Theorem]
\label{thm:pruning}
Let $N = (P,T,F)$ be an IO net, let $M$ be a marking of $N$, and let
$M'' \trans{*} M'$ be a firing sequence of $N$ such that $M' \geq M$.
There exist markings $S''$ and $S'$ such that 
\begin{center}
\(
\begin{array}[b]{@{}c@{}c@{}c@{}c@{}c@{}c@{}c@{}}
M'' &  \trans{\hspace{1em}*\hspace{1em}} & M' & \  \geq \ &M  \\[0.1cm]
\geq &  & \geq \\[0.1cm]
S'' & \trans{\hspace{1em}*\hspace{1em}} &S'  & \geq  & M
\end{array}
\)
\end{center}
and $|S''| \leq |M| + |P|^3$.
\end{theorem}

It is easy to see that for $M' = M$ the Pruning Theorem holds, because since $N$ is conservative, $|M''|=|M|$ and we can choose $S''=M''$.
It is also not difficult to find a non-IO net for which the theorem does not hold. 

\begin{example}
Consider our IO net represented in Figure \ref{figure-running-example}.
There is a firing sequence $(30,0,1) \trans{*} (0,0,31)$ where $(0,0,31)$ covers $(0,0,2)$. 
By application of the Pruning Theorem, we obtain a firing sequence $(3,0,1) \trans{*} (0,0,4) \geq (0,0,2)$ where $|(3,0,1)|=4$ is smaller than $|(0,0,2)|+3^3=29$.
\end{example}

\begin{example}[A non-IO net for which the theorem does not hold.] 
To see that the IO condition cannot be replaced with conservativeness of the network, consider the net with $4$ places $q_1,q_2,q_3,q_4$ and a single transition $(q_1,q_2)\mapsto(q_3,q_4)$.
There is a firing sequence $(1000,1000,0,0) \trans{*} (0,0,1000,1000) \geq (0,0,100,0)$.
But to cover $(0,0,100,0)$ from a marking below $(1000,1000,0,0)$ we need to fire the transition at least $100$ times.
This requires a marking with $200 > 100 + 4^3$ tokens.
\end{example}

\subsection{Trajectories, Histories, Realizability}

Since the transitions of IO nets do not create or destroy tokens, we can give tokens identities.
Given a firing sequence, each token of the initial marking follows a \emph{trajectory}, or sequence of \emph{steps}, through the places of the net until it reaches the final marking of the sequence.

\begin{definition}
A \emph{trajectory} is a sequence $\tau =p_1 \ldots p_n$ of places. We denote $\tau(i)$ the $j$-th place of $\tau$. The \emph{$j$-th step} of $\tau$ is the pair $\tau(i)\tau(i+1)$ of adjacent places.

A \emph{history} is a multiset of trajectories of the same length. The length of a history is the common length of its trajectories. Given a history $H$ of length $n$ and index $1 \leq i \leq n$, \emph{the $j$-th marking of $H$}, denoted $M_{H}^i$, is defined as follows: for every place $p$, $M_{H}^i(p)$ is the number of trajectories $\tau \in H$ such that $\tau(i)=p$. The markings 
$M_{H}^1$ and $M_{H}^n$ are called the \emph{initial} and \emph{final} markings of $H$.

A history $H$ of length $n\geq 1$ is \emph{realizable} in an IO net $N$ if there exist transitions of $N$ $t_1, \ldots, t_{n-1}$ 
and numbers $k_1, \ldots, k_{n-1} \geq 0$ such that $M_{H}^1 \trans{t_1^{k_1}}M_{H}^2 \cdots  M_{H}^{n-1} \trans{t_{n-1}^{k_{n-1}}} M_{H}^n$.
Notice that a history of length $1$ is always realizable.
\end{definition}

\begin{remark}
Notice that there may be more than one realizable history corresponding to a firing sequence in an IO net, because the firing sequence does not keep track of which token goes where, while the history does.
\end{remark}

\vspace{-0.5cm}
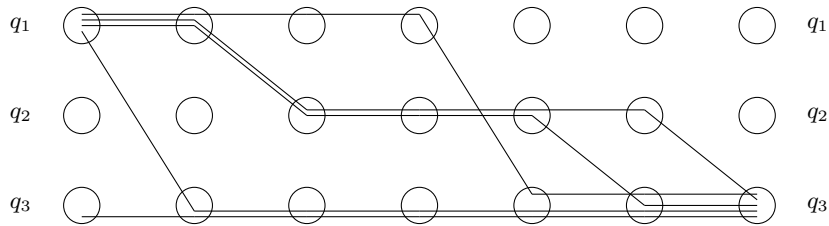
\begin{figure}
\centering
\begin{tikzpicture}
  \node[circle, draw]
    (Node11)
    {~~~}; 
    \node[right=1 of Node11, circle, draw]
    (Node12)
    {~~~};
    \node[right=1 of Node12, circle, draw]
    (Node13)
    {~~~};
    \node[right=1 of Node13, circle, draw]
    (Node14)
    {~~~};
    \node[right=1 of Node14, circle, draw]
    (Node15)
    {~~~};
    \node[right=1 of Node15, circle, draw]
    (Node16)
    {~~~};
    \node[right=1 of Node16, circle, draw]
    (Node17)
    {~~~};
    \node[below=0.7 of Node11, circle, draw]
    (Node21)
    {~~~};
    \node[right=1 of Node21, circle, draw]
    (Node22)
    {~~~};
    \node[right=1 of Node22, circle, draw]
    (Node23)
    {~~~};
    \node[right=1 of Node23, circle, draw]
    (Node24)
    {~~~};
    \node[right=1 of Node24, circle, draw]
    (Node25)
    {~~~};
    \node[right=1 of Node25, circle, draw]
    (Node26)
    {~~~};
    \node[right=1 of Node26, circle, draw]
    (Node27)
    {~~~};
    \node[below=0.7 of Node21, circle, draw]
    (Node31)
    {~~~};
    \node[right=1 of Node31, circle, draw]
    (Node32)
    {~~~};
    \node[right=1 of Node32, circle, draw]
    (Node33)
    {~~~}; 
    \node[right=1 of Node33, circle, draw]
    (Node34)
    {~~~};
    \node[right=1 of Node34, circle, draw]
    (Node35)
    {~~~};
    \node[right=1 of Node35, circle, draw]
    (Node36)
    {~~~};
    \node[right=1 of Node36, circle, draw]
    (Node37)
    {~~~};
        \node[left=0.3 of Node11]
        (Node11l)
        {$q_1$};
        \node[right=0.3 of Node17]
        (Node17r)
        {$q_1$};
        \node[left=0.3 of Node21]
        (Node21l)
        {$q_2$};
        \node[right=0.3 of Node27]
        (Node27r)
        {$q_2$};
        \node[left=0.3 of Node31]
        (Node31l)
        {$q_3$};
        \node[right=0.3 of Node37]
        (Node37r)
        {$q_3$};
    
  \draw
  ($(Node11.center)+(0,1.50mm)$)
    --
  ($(Node12.center)+(0,1.50mm)$)
  ;
  \draw
  ($(Node12.center)+(0,1.50mm)$)
    --
  ($(Node13.center)+(0,1.50mm)$)
  ;
  \draw($(Node13.center)+(0,1.50mm)$)
    --
  ($(Node14.center)+(0,1.50mm)$)
  ;
  \draw
  ($(Node14.center)+(0,1.50mm)$)
    --
  ($(Node35.center)+(0,1.50mm)$)
  ;
  \draw
  ($(Node35.center)+(0,1.50mm)$)
    --
  ($(Node36.center)+(0,1.50mm)$)
  ;
  \draw
  ($(Node36.center)+(0,1.50mm)$)
    --
  ($(Node37.center)+(0,1.50mm)$)
  ;
  \draw
  ($(Node11.center)+(0,0.0mm)$)
    --
  ($(Node12.center)+(0,0.0mm)$)
  ;
  \draw
  ($(Node12.center)+(0,0.0mm)$)
    --
  ($(Node23.center)+(0,0.0mm)$)
  ;
  \draw
  ($(Node23.center)+(0,0.0mm)$)
    --
  ($(Node24.center)+(0,0.0mm)$)
  ;
  \draw
  ($(Node24.center)+(0,0.0mm)$)
    --
  ($(Node25.center)+(0,0.0mm)$)
  ;
  \draw
  ($(Node25.center)+(0,0.0mm)$)
    --
  ($(Node36.center)+(0,0.0mm)$)
  ;
  \draw
  ($(Node36.center)+(0,0.0mm)$)
    --
  ($(Node37.center)+(0,0.0mm)$)
  ;
  \draw
  ($(Node11.center)+(0,0.75mm)$)
    --
  ($(Node12.center)+(0,0.75mm)$)
  ;
  \draw
  ($(Node12.center)+(0,0.75mm)$)
    --
  ($(Node23.center)+(0,0.75mm)$)
  ;
  \draw
  ($(Node23.center)+(0,0.75mm)$)
    --
  ($(Node24.center)+(0,0.75mm)$)
  ;
  \draw
  ($(Node24.center)+(0,0.75mm)$)
    --
  ($(Node25.center)+(0,0.75mm)$)
  ;
  \draw
  ($(Node25.center)+(0,0.75mm)$)
    --
  ($(Node26.center)+(0,0.75mm)$)
  ;
  \draw
  ($(Node26.center)+(0,0.75mm)$)
    --
  ($(Node37.center)+(0,0.75mm)$)
  ;
  \draw
  ($(Node11.center)+(0,-0.75mm)$)
    --
  ($(Node32.center)+(0,-0.75mm)$)
  ;
  \draw
  ($(Node32.center)+(0,-0.75mm)$)
    --
  ($(Node33.center)+(0,-0.75mm)$)
  ;
  \draw
  ($(Node33.center)+(0,-0.75mm)$)
    --
  ($(Node34.center)+(0,-0.75mm)$)
  ;
  \draw
  ($(Node34.center)+(0,-0.75mm)$)
    --
  ($(Node35.center)+(0,-0.75mm)$)
  ;
  \draw
  ($(Node35.center)+(0,-0.75mm)$)
    --
  ($(Node36.center)+(0,-0.75mm)$)
  ;
  \draw
  ($(Node36.center)+(0,-0.75mm)$)
    --
  ($(Node37.center)+(0,-0.75mm)$)
  ;
  \draw
  ($(Node31.center)+(0,-1.50mm)$)
    --
  ($(Node32.center)+(0,-1.50mm)$)
  ;
  \draw
  ($(Node32.center)+(0,-1.50mm)$)
    --
  ($(Node33.center)+(0,-1.50mm)$)
  ;
  \draw
  ($(Node33.center)+(0,-1.50mm)$)
    --
  ($(Node34.center)+(0,-1.50mm)$)
  ;
  \draw
  ($(Node34.center)+(0,-1.50mm)$)
    --
  ($(Node35.center)+(0,-1.50mm)$)
  ;
  \draw
  ($(Node35.center)+(0,-1.50mm)$)
    --
  ($(Node36.center)+(0,-1.50mm)$)
  ;
  \draw
  ($(Node36.center)+(0,-1.50mm)$)
    --
  ($(Node37.center)+(0,-1.50mm)$)
  ;

\end{tikzpicture}
\caption{ Realizable history in our IO net with three states.}
\label{figure-history}
\end{figure}

\vspace{-0.5cm}
\begin{example}
\label{history-ex}
Histories can be graphically represented. 
Consider Figure \ref{figure-history} which illustrates a history $H$ of length $7$.
It consists of five trajectories: one trajectory from $q_3$ to $q_3$ passing only through $q_3$, and four trajectories from $q_1$ to $q_3$ which follow different place sequences.
$H$'s first marking is $M_{H}^1= (4,0,1)$ and $H$'s seventh and last marking is $M_{H}^7=(0,0,5)$.
History $H$ is realizable in the IO net N of Figure \ref{figure-running-example} 
which has place set $\{ q_1,q_2,q_3 \}$ and transitions $t_1 = (q_1, q_1)\mapsto(q_2,q_1), t_2 = (q_2,q_2)\mapsto(q_3,q_2),t_3 = (q_1,q_3)\mapsto(q_3,q_3)$ and $t_4 = (q_2,q_3)\mapsto(q_3,q_3)$.
Indeed $M_{H}^1 \trans{t_3 t_1^2 t_3 t_2 t_4 } M_{H}^7$.
\end{example}

We define a class of histories sufficient for describing all the firing sequences for IO nets.

\begin{definition}
A step $\tau(i)\tau(i+1)$ of a trajectory $\tau$ is \emph{horizontal} if $\tau(i) = \tau(i+1)$, and \emph{non-horizontal} otherwise.

A history $H$ of length $n$ is \emph{well-structured} if for every $1 \leq i \leq n-1$ one of the two following conditions hold:
\begin{itemize}
\item For every trajectory $\tau \in H$, the $j$-th step of $\tau$ is horizontal.
\item For every two trajectories $\tau_1, \tau_2 \in H$, if the $j$-th steps of $\tau_1$ and $\tau_2$ are non-horizontal, then they are equal.
\end{itemize}
\end{definition}

We then have the following result, whose proof can be found in the full version of this paper.

\begin{restatable}{lemma}{LemmaRealizable}
\label{lem:realizable}
Let $N$ be an IO net.
Then $M \trans{*} M'$ if{}f there exists a well-structured history realizable in $N$ with $M$ and $M'$ as initial and final markings.
\end{restatable}

We now proceed to give a syntactic characterization of the well-structured realizable histories. 

\begin{definition}
$H$ is \emph{compatible with $N$}
if for every trajectory $\tau$ of $H$ and for every non-horizontal step $\tau(i)\tau(i+1)$ of $\tau$, 
the net $N$ contains a transition $(\tau(i), p_o)~\mapsto~(\tau(i+1), p_o)$ for some place $p_o$
and $H$ contains a trajectory $\tau'$ with $\tau'(i)=\tau'(i+1)=p_o$.
\end{definition}

\begin{restatable}{lemma}{LemmaHistory}
\label{lem:compatible}
Let $N$ be an IO net. A well-structured history is realizable in $N$ if{}f it is compatible with $N$.
\end{restatable}

\begin{example}
In the realizable history $H$ of Figure \ref{figure-history}, all the trajectories are such that the third step is horizontal. 
For every step except the third, all the non-horizontal steps are equal, so $H$ is well-structured.
For $N$ the IO net of Figure \ref{figure-running-example}, $H$ is indeed compatible with $N$.
\end{example}

\subsection{Pruning Histories}

We start by introducing bunches of trajectories.

\begin{definition}
 A \emph{bunch} is a multiset of trajectories with the same length and the same
        initial and final place. 
\end{definition}

\begin{example}
Figure \ref{figure-history}'s realizable history is constituted of a trajectory from $q_3$ to $q_3$ and a bunch $B$ with initial place $q_1$ and final place $q_3$ made up of four different trajectories. 
\end{example}

We show that every well-structured realizable history containing a bunch of size larger than $|P|$ can be ``pruned'', meaning that the bunch can be replaced by a smaller one, while keeping the history well-structured and realizable.

\begin{lemma}
\label{lm:pruning}
Let $N$ be an IO net. Let $H$ be a well-structured history realizable in $N$ containing a bunch $B\subseteq H$ of size larger than $|P|$. There exists a nonempty bunch $B'$ of size at most $|P|$ with the same initial and final
places as $B$, such that the history $H' \defeq H - B + B'$ (where $+, -$ denote multiset addition and subtraction) is also well-structured and realizable in $N$.
\end{lemma}
\begin{proof}
Let $P_B$ be a set of all places visited by at least one trajectory in the bunch $B$.
        For every $p\in P_B$ let $f(p)$ and $l(p)$ be the earliest and the latest
        moment in time when this place has been used by any of the trajectories
        (the first and the last occurrence can be in different trajectories).

        Let $\tau_p, p\in P_B$ be a trajectory that first goes to $p$ by the moment $f(p)$, then
        waits there until $l(p)$, then goes from $p$ to the final place.
        To go to and from $p$ it uses fragments of trajectories of $B$.

        We will take $B'=\{\tau_p\mid p\in P_B\}$ and prove that replacing $B$ with $B'$ in $H$
        does not violate the requirements for being a well-structured history realizable in $N$.
        Note that we can copy the same fragment of
        a trajectory multiple times.

        First let us check the well-structuring condition. Note that we build $\tau_p$ by taking fragments
        of existing trajectories and using them at the exact same moments in time, and by adding some
        horizontal fragments. Therefore, the set of non-horizontal steps in $B'$ is a subset (if
        we ignore multiplicity) of the set of non-horizontal steps in $B$, and the replacement
        operation cannot increase the set of non-horizontal steps occurring in $H$.

        Now let us check compatibility with $N$. 
        Consider any non-horizontal step in $H'$ in any trajectory at position $(i,i+1)$. By construction, the same step at the same position is also present in $H$. History $H$ is realizable in $N$ and thus by Lemma \ref{lem:compatible} it is compatible with $N$, so $H$ contains an enabling horizontal step $p_o p_o$ in some trajectory at that position $(i,i+1)$. There are two cases: either that step $p_o p_o$ was provided by a bunch being pruned, or by a bunch not affected by pruning. In the first case, note that the place $p_o$ of this horizontal step must be first observed no later than $j$, and last observed not earlier than $j+1$. This implies $f(p_o)\leq i<i+1\leq l(p_o)$. As $H'$ contains a horizontal step $p_o p_o$ for all positions between $f(p_o)$ and $l(p_o)$, in particular it contains it at position $(i,i+1)$. In the second case the same horizontal step is present in $H'$ as a part of the same trajectory. 
       
       So $H'$ is well-structured and compatible with $N$, and thus by Lemma \ref{lem:compatible} realizable in $N$.
        \qed
\end{proof}

\vspace{-0.5cm}
\begin{figure}
\vskip-0.5cm
\centering
\begin{tikzpicture}

  \node[circle, draw]
    (Node11)
    {~~~}; 
    \node[right=1 of Node11, circle, draw]
    (Node12)
    {~~~};
    \node[right=1 of Node12, circle, draw]
    (Node13)
    {~~~};
    \node[right=1 of Node13, circle, draw]
    (Node14)
    {~~~};
    \node[right=1 of Node14, circle, draw]
    (Node15)
    {~~~};
    \node[right=1 of Node15, circle, draw]
    (Node16)
    {~~~};
    \node[right=1 of Node16, circle, draw]
    (Node17)
    {~~~};
    \node[below=0.7 of Node11, circle, draw]
    (Node21)
    {~~~};
    \node[right=1 of Node21, circle, draw]
    (Node22)
    {~~~};
    \node[right=1 of Node22, circle, draw]
    (Node23)
    {~~~};
    \node[right=1 of Node23, circle, draw]
    (Node24)
    {~~~};
    \node[right=1 of Node24, circle, draw]
    (Node25)
    {~~~};
    \node[right=1 of Node25, circle, draw]
    (Node26)
    {~~~};
    \node[right=1 of Node26, circle, draw]
    (Node27)
    {~~~};
    \node[below=0.7 of Node21, circle, draw]
    (Node31)
    {~~~};
    \node[right=1 of Node31, circle, draw]
    (Node32)
    {~~~};
    \node[right=1 of Node32, circle, draw]
    (Node33)
    {~~~}; 
    \node[right=1 of Node33, circle, draw]
    (Node34)
    {~~~};
    \node[right=1 of Node34, circle, draw]
    (Node35)
    {~~~};
    \node[right=1 of Node35, circle, draw]
    (Node36)
    {~~~};
    \node[right=1 of Node36, circle, draw]
    (Node37)
    {~~~};
        \node[left=0.3 of Node11]
        (Node11l)
        {$q_1$};
        \node[right=0.3 of Node17]
        (Node17r)
        {$q_1$};
        \node[left=0.3 of Node21]
        (Node21l)
        {$q_2$};
        \node[right=0.3 of Node27]
        (Node27r)
        {$q_2$};
        \node[left=0.3 of Node31]
        (Node31l)
        {$q_3$};
        \node[right=0.3 of Node37]
        (Node37r)
        {$q_3$};

        \node[above=0.0 of Node11, color=blue](Node11u){$f$};
        \node[above=0.0 of Node14, color=blue](Node14u){$l$};
        \node[above=0.0 of Node23, color=blue](Node23u){$f$};
        \node[above=0.0 of Node26, color=blue](Node26u){$l$};
        \node[above=0.0 of Node32, color=blue](Node32u){$f$};
        \node[above=0.0 of Node37, color=blue](Node37u){$l$};
        
  \draw[blue]
  ($(Node11.center)+(0,1.50mm)$)
    --
  ($(Node12.center)+(0,1.50mm)$)
  ;
  \draw[blue]
  ($(Node12.center)+(0,1.50mm)$)
    --
  ($(Node13.center)+(0,1.50mm)$)
  ;
  \draw[blue]($(Node13.center)+(0,1.50mm)$)
    --
  ($(Node14.center)+(0,1.50mm)$)
  ;
  \draw[blue]
  ($(Node14.center)+(0,1.50mm)$)
    --
  ($(Node35.center)+(0,1.50mm)$)
  ;
  \draw[blue]
  ($(Node35.center)+(0,1.50mm)$)
    --
  ($(Node36.center)+(0,1.50mm)$)
  ;
  \draw[blue]
  ($(Node36.center)+(0,1.50mm)$)
    --
  ($(Node37.center)+(0,1.50mm)$)
  ;
  \draw[blue]
  ($(Node11.center)+(0,0.0mm)$)
    --
  ($(Node12.center)+(0,0.0mm)$)
  ;
  \draw[blue]
  ($(Node12.center)+(0,0.0mm)$)
    --
  ($(Node23.center)+(0,0.0mm)$)
  ;
  \draw[blue]
  ($(Node23.center)+(0,0.0mm)$)
    --
  ($(Node24.center)+(0,0.0mm)$)
  ;
  \draw[blue]
  ($(Node24.center)+(0,0.0mm)$)
    --
  ($(Node25.center)+(0,0.0mm)$)
  ;
  \draw[blue]
  ($(Node25.center)+(0,0.0mm)$)
    --
  ($(Node26.center)+(0,0.0mm)$)
  ;
  \draw[blue]
  ($(Node26.center)+(0,0.0mm)$)
    --
  ($(Node37.center)+(0,0.0mm)$)
  ;
  \draw[blue]
  ($(Node11.center)+(0,-0.75mm)$)
    --
  ($(Node32.center)+(0,-0.75mm)$)
  ;
  \draw[blue]
  ($(Node32.center)+(0,-0.75mm)$)
    --
  ($(Node33.center)+(0,-0.75mm)$)
  ;
  \draw[blue]
  ($(Node33.center)+(0,-0.75mm)$)
    --
  ($(Node34.center)+(0,-0.75mm)$)
  ;
  \draw[blue]
  ($(Node34.center)+(0,-0.75mm)$)
    --
  ($(Node35.center)+(0,-0.75mm)$)
  ;
  \draw[blue]
  ($(Node35.center)+(0,-0.75mm)$)
    --
  ($(Node36.center)+(0,-0.75mm)$)
  ;
  \draw[blue]
  ($(Node36.center)+(0,-0.75mm)$)
    --
  ($(Node37.center)+(0,-0.75mm)$)
  ;
  \draw
  ($(Node31.center)+(0,-1.50mm)$)
    --
  ($(Node32.center)+(0,-1.50mm)$)
  ;
  \draw
  ($(Node32.center)+(0,-1.50mm)$)
    --
  ($(Node33.center)+(0,-1.50mm)$)
  ;
  \draw
  ($(Node33.center)+(0,-1.50mm)$)
    --
  ($(Node34.center)+(0,-1.50mm)$)
  ;
  \draw
  ($(Node34.center)+(0,-1.50mm)$)
    --
  ($(Node35.center)+(0,-1.50mm)$)
  ;
  \draw
  ($(Node35.center)+(0,-1.50mm)$)
    --
  ($(Node36.center)+(0,-1.50mm)$)
  ;
  \draw
  ($(Node36.center)+(0,-1.50mm)$)
    --
  ($(Node37.center)+(0,-1.50mm)$)
  ;

\end{tikzpicture}
\caption{ History $H$ of Figure \ref{figure-history} after pruning.}
\vskip-0.5cm
\label{figure-pruned}
\end{figure}
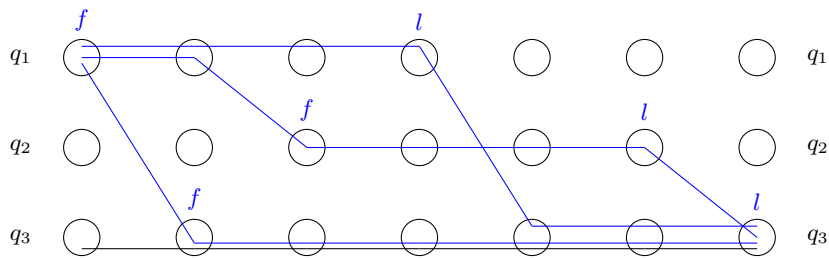

\vspace{-0.5cm}
\begin{example}
Consider the well-structured realizable history of Figure \ref{figure-history}, leading from $(4,0,1)$ to $(0,0,5)$, which covers marking $(0,0,2)$.
Bunch $B$ from $q_1$ to $q_3$ is of size four which is bigger than $|P|=3$. 
The set $P_B$ of places visited by trajectories of $B$ is equal to $P$.
Figure \ref{figure-pruned} is annotated with the first and last moments $f(p)$ and $l(p)$ for $p \in P_B$.
Lemma \ref{lm:pruning} applied to $H$ and $B$ ``prunes" bunch $B$ into $B'$ made up of trajectories $\tau_{q_1},\tau_{q_2},\tau_{q_3}$, drawn in blue in Figure \ref{figure-pruned}.
Notice that in this example, the non-horizontal $5$-th step in $H$ does not appear in the new well-structured and realizable history $H' = H - B + B'$.
History $H'$ is such that $M_{H'}^1=(3,0,1) \trans{t_3 t_1 t_3 t_4 }(0,0,4)=M_{H'}^7$ and $M_{H'}^7 \geq (0,0,2)$. 
\end{example}

\subsection{Proof of the Pruning Theorem}
\label{subsec-proof}

Using Lemma \ref{lm:pruning} we can now finally prove the Pruning Theorem:

\begin{proof}[of Theorem \ref{thm:pruning}]
Let $M'' \trans{*}  M'  \  \geq  \ M$. 
By Lemma \ref{lem:realizable}, there is a well-structured realizable history $H$ with $M''$ and $M'$ as initial and final markings, respectively. Let $H_M \subset H$ be an arbitrary sub(multi)set of $H$ with final
marking $M$, and define $H'=H-H_M$. Further, for every $p, p' \in P$, let $H'_{p,p'}$ be the bunch of all trajectories of $H'$ with $p$ and $p'$ as initial and final places, respectively. We have
$$H' = \sum_{p,p' \in P} H'_{p,p'}$$

So $H'$ is the union of $|P|^2$ (possibly empty) bunches. Apply Lemma \ref{lm:pruning} to each bunch of $H'$ with more than $|P|$ trajectories yields a new history

$$H'' = \sum_{p,p' \in p} H''_{p,p'}$$

\noindent such that $|H''_{p,p'}| \leq |P|$ for every $p, p' \in P$, and such that the history $H'' + H_M$ is well-structured and realizable. 

Let $S''$ and $S'$ be the initial and final markings of $H'' + H_M$. We show that $S''$ and $S'$ satisfy the required properties:
\begin{itemize}
\item $S'' \trans{*} S'$, because  $H'' + H_M$ is well-structured and realizable.
\item $S' \geq M$, because $H_M \subseteq H'' + H_M$. 
\item $|S''| \leq |M| + |P|^3$ because $|H'' + H_M| = \sum_{p', p} |H''_{p,p'}| + |H_M| \leq |P|^2 \cdot |P| + |M| = |M| + |P|^3$.
\end{itemize}
This concludes the proof.
\qed
\end{proof}

\begin{remark}
A slight modification of our construction allows one to prove Theorem~\ref{thm:pruning} (but not Lemma \ref{lm:pruning}) with $2|P|^2$ overhead instead of $|P|^3$. We provide more details in the full version \cite{EsparzaRW19}.
However, since some results of Section \ref{CC} explicitly rely on Lemma~\ref{lm:pruning}, we prove Theorem~\ref{thm:pruning} as a consequence of Lemma~\ref{lm:pruning} for simplicity.
\end{remark}
\section{Counting Constraints and Counting Sets}
\label{CC}
In this section we first briefly recall \emph{counting constraints} \cite{EsparzaGMW18}
\footnote{Actually, our counting constraints correspond to the
``counting constraints in normal form'' of [14]. We shorten the name,
because we never need counting constraints not in normal form.}, a class of constraints that allow us to finitely represent  (possibly infinite) sets of markings, called \emph{counting sets}. We prove Theorem \ref{thm:new18}, a powerful result stating that counting sets of IO nets are closed under reachability, and giving a very tight relation between
the sizes of the constraints representing a counting set, and the set of markings reachable from it. Theorem \ref{thm:new18} strongly improves on Theorem 18 of ~\cite{EsparzaGMW18}.

\paragraph{Counting constraints and counting sets. }  Recall Definition \ref{def:cube} which defines a \emph{cube} of a net $N$ as a set of markings given by a lower bound $L \colon P \rightarrow \N$ and an upper bound $U \colon P \rightarrow \N \cup \{\infty\}$, written $(L,U)$, and such that $M \in (L,U)$ if{}f $L \leq M \leq U$. 
In the rest of the paper, the term cube will refer both to the set of markings and to the description by upper and lower bound $(L,U)$.
A \emph{counting constraint} is a formal finite union of cubes, i.e. a formal finite union of upper and lower bound pairs of the form $(L,U)$.
The semantics of a counting constraint is called a \emph{counting set} and it is the union of the cubes defining the counting constraint.
The counting set for a counting constraint $\Gamma$ is denoted $\sem{\Gamma}$.
Notice that one counting set can be the semantics of different counting constraints.
For example, consider a net with just one place $p_1$. Let $(L,U)=(1,3)$, $(L',U')=(2,4)$, $(L'',U'')=(1,4)$. 
The counting constraints $(L,U)\cup(L',U')$ and $(L'',U'')$ define the same counting set.
It is easy to show (see also \cite{EsparzaGMW18}) that counting sets are closed under Boolean operations.

\paragraph{Measures of counting constraints.} Let $C=(L,U)$ be a cube, and let $\Gamma = \bigcup_{i=1}^m C_i$
be a counting constraint. We use the following notations:
$$\begin{array}{lcl}
\lnorm{C} \defeq \displaystyle \sum_{p \in P} L(p) & \quad & \unorm{C} \defeq \displaystyle \sum_{\substack{p\in P\\ U(p)<\infty}} \!\!\!\!U(p) \mbox{ (and \( 0\) if \(U(p)=\infty\) for all \(p\))}. \\[0.2cm]
\lnorm{\Gamma} \defeq \displaystyle \max_{i\in [1,m]} \{ \lnorm{C_i} \} &  & \unorm{\Gamma} \defeq \displaystyle \max_{i\in[1,m]} \{ \unorm{C_i} \}
\end{array}$$
\noindent We call $\lnorm{C}$ the $L$-norm  and $\unorm{C}$ the $U$-norm of $C$. Similarly for $\Gamma$.
We recall Proposition 5 of~\cite{EsparzaGMW18} for the norms of the union, intersection and complement.

\begin{proposition}
\label{prop:oponconf}
Let $\Gamma_1, \Gamma_2$ be counting constraints.
\begin{itemize}
\item There exists a counting constraint $\Gamma$ with $\sem{\Gamma} = \sem{\Gamma_1} \cup \sem{\Gamma_2}$ such that
$\unorm{\Gamma} \leq \max \{\unorm{\Gamma_1}, \unorm{\Gamma_2} \}$ and $\lnorm{\Gamma} \leq \max \{\lnorm{\Gamma_1}, \lnorm{\Gamma_2} \}$.
\item  There exists a counting constraint $\Gamma$ with $\sem{\Gamma} = \sem{\Gamma_1} \cap \sem{\Gamma_2}$ such that
$\unorm{\Gamma} \leq \unorm{\Gamma_1} + \unorm{\Gamma_2}$ and $\lnorm{\Gamma} \leq \lnorm{\Gamma_1} + \lnorm{\Gamma_2}$.
\item There exists a counting constraint $\Gamma$ with $\sem{\Gamma} = \N^n \setminus \sem{\Gamma_1}$ such that
$\unorm{\Gamma} \leq n\lnorm{\Gamma_1}$ and $\lnorm{\Gamma} \leq n\unorm{\Gamma_1} + n$.
\end{itemize}
 \end{proposition}

\paragraph{Predecessors and successors of counting sets.}
Fix an IO net $N=(P,T,F)$. The sets of predecessors and successors of a set $\mathcal{M}$ of markings of $N$ are defined as follows:
$\pre^*(\mathcal{M}) \defeq \{ M' | \exists M \in \mathcal{M} \, . \, M' \xrightarrow{*} M \}$, and 
$\post^*(\mathcal{M}) \defeq \{ M | \exists M' \in \mathcal{M} \, . \, M' \xrightarrow{*} M \}$.

\begin{lemma}
\label{lm:smallminterm}
Let $(L,U)$ be a cube of an IO net $N$ of place set $P$.
For all $M' \in \pre^*(L,U)$, there exists a cube $(L',U')$ such that 
\begin{enumerate}
\item $M' \in (L',U') \subseteq  \pre^*(L,U)$, and
\item $\lnorm{(L',U')} \leq \lnorm{(L,U)}  + |P|^3$ and $\unorm{(L',U')} \leq \unorm{(L,U)}$.
\end{enumerate}
\end{lemma}
\begin{proof}
Let $M'$ be a marking of $\pre^*(L,U)$.
There exists a marking $M \in (L,U)$ such that $M' \longrightarrow M$, and $M \geq L$.
The construction from the Pruning Theorem 
applied to this firing sequence
yields markings $S',S$ such that 
\begin{center}
\(
\begin{array}[b]{@{}c@{}c@{}c@{}c@{}c@{}c@{}c@{}}
M' &  \trans{\hspace{1em}*\hspace{1em}} & M & \  \geq \ &L  \\[0.1cm]
\geq &  & \geq \\[0.1cm]
S' & \trans{\hspace{1em}*\hspace{1em}} &S  & \geq & L
\end{array}
\)
\end{center}
and $|S'| \leq |L| + |P|^3$.
Since $M$ is in $(L,U)$, we have $U\geq M \geq S \geq L$ and so marking $S$ is in $(L,U)$ and $S'$ is in $\pre^*(L,U)$.

We want to find $L',U'$ satisfying the conditions of the Lemma, i.e. such that $M' \in (L', U')$ and $(L',U') \subseteq  \pre^*(L,U)$.
We define $L'$ as equal to marking $S'$ over each place of $P$.
The following part of the proof plays out in the setting of the Pruning Theorem section, in which the tokens are de-anonymized. 
Let $H_M$ be a well-structured realizable history from $M'$ to $M$. 
Let $p$ be a place of $P$.
We want to define $U'(p)$.
Consider $\mathcal{B}^M_p$ the set of bunches in history $H_M$ that have $p$ as an initial place.
For every bunch $B$, let $f_B$ be the final place of the bunch.
We define $U'(p)$ depending on the final places of bunches in $\mathcal{B}^M_p$.

\noindent
\emph{Case 1.} 
There exists a bunch $B$ in $\mathcal{B}^M_p$ whose final place $f_B$ is such that $U(f_B) = \infty$.
In this case we define $U'(p)$ to be $\infty$.

\noindent
\emph{Case 2.}
For all bunches $B$ in $\mathcal{B}^M_p$, the final place $f_B$ of $B$ is such that $U(f_B) < \infty$.
In this case we define $U'(p)$ to be $\sum_{B \in \mathcal{B}^M_p} size(B)$, where $size(B)$ is the number of trajectories with multiplicity in $B$, and $0$ if $\mathcal{B}^M_p$ is empty.

Let us show that $(L',U')$ has the properties we want.
The number of tokens in marking $M'$ at place $p \in P$ is the sum of the sizes of the bunches that start from $p$ in history $H_M$.
That is, $M'(p)=\sum_{B \in \mathcal{B}^M_p} size(B)$ which is exactly $U'(p)$ when $U'(p)$ is finite.
Thus for all $p \in P$, $M'(p) \leq U'(p)$ and $M'(p) \geq S'(p)=L'(p)$, so $M'$ is in $(L', U')$.

The construction from the Pruning Theorem ``prunes" history $H_M$ from $M'$ to $M$ into a well-structured realizable history $H_S$ from $S'$ to $S$ with the same set of non-empty bunches.
We are going to show that $(L',U') \subseteq \pre^*(L,U)$ by ``boosting" the bunches of history $H_S$ to create histories $H_R$ which will start in any marking $R'$ of $(L',U')$ and end at some marking $R$ in $(L,U)$.
For any constant $k \in \N$, a bunch $B$ of history $H_S$ is \emph{boosted by k} into a bunch $B'$ by selecting any trajectory $\tau$ in $B$ and augmenting its multiplicity by $k$ to create a new bunch $B'$ of size $size(B) + k$.

\begin{figure}
\vskip-0.5cm
\centering
\begin{tikzpicture}

\node[draw=none,fill=none] (Mp) {$M'$};%
\node[right=4 of Mp, draw=none,fill=none] (M) {$M$};%
\node[below=0.45 of Mp, draw=none,fill=none] (Np) {$R'$};%
\node[right=4 of Np, draw=none,fill=none] (N) {$R$};%
\node[below=0.45 of Np, draw=none,fill=none] (Sp) {$S'$};%
\node[right=4 of Sp, draw=none,fill=none] (S) {$S$};%
\node[below=0.0 of Mp, draw=none,fill=none] (small) {$\geq$};%
\node[below=0.0 of Np, draw=none,fill=none] (smal) {$\geq$};%
\node[below=0.0 of M, draw=none,fill=none] (sma) {$\geq$};%
\node[below=0.0 of N, draw=none,fill=none] (sm) {$\geq$};%
\node[right=0.1 of M, draw=none,fill=none] (U) {$\leq U$};%
\node[right=0.1 of S, draw=none,fill=none] (L) {$\geq L$};%
\node[left=0.1 of Mp, draw=none,fill=none] (Up) {$U'\geq$};%
\node[left=0.1 of Sp, draw=none,fill=none] (Lp) {$L'\leq$};%
%

\draw[->] (Mp) to node[below] {$H_M$} (M);%
\draw[->] (Np) to node[below] {$H_R$} (N);%
\draw[->] (Sp) to node[below] {$H_S$} (S);%

\end{tikzpicture}
\caption{}
\label{figure-boost}
\vskip-0.5cm
\end{figure}

Let $R'$ be a marking in $(L',U')$.
We construct a new history $H_R$ starting in $R'$, and we prove that its final place is in $(L,U)$.
What we aim to build is illustrated in Figure \ref{figure-boost}.
We initialize $H_R$ as the bunches of history $H_S$.
We call $\mathcal{B}^S_p$ the set of the bunches of $H_S$ starting in $p$.

For $p$ such that there is a bunch $B_S \in \mathcal{B}^S_p$ with infinite $U(f_{B_S})$, i.e. such that $\mathcal{B}^S_p$ is in \emph{Case 1} defined above, we take this bunch $B_S$ and boost it by $R'(p) - S'(p)$ into a new bunch $B_R$. 
Informally, we need not worry about exceeding the bound $U$ on the final place of the trajectories of $B_R$, because this place is $f_{B_S}$ and its upper bound is infinite.
The number of trajectories starting in $p$ in history $H_R$ is now $R'(p)$.

Otherwise, $p$ is such that $\mathcal{B}^S_p$ is in \emph{Case 2}, so we know that $R'(p)\leq M'(p)$ because $U'(p)$ was defined to be $M'(p)$. 
Each bunch in $\mathcal{B}^S_p$ in history $H_S$ has a corresponding bunch in history $H_M$ because the pruning operation never erases a bunch completely, it only diminishes its size.
We can boost all bunches in $\mathcal{B}^S_p$ to the size of the corresponding bunches in $H_M$ and not exceed the finite bounds of $U$ on the final places of these bunches.
We arbitrarily select bunches in $\mathcal{B}^S_p$ which we boost so that the sum of the size of bunches in $\mathcal{B}^S_p$ is equal to $R'(p)$.

Now by construction, history $H_R$ starts in marking $R'$, and it ends in a marking $R$ such that $S\leq R\leq U$, as every bunch is either boosted to a size no greater than it had in $H_M$, or leads to a place $p$ with $U(p)=\infty$.
Since $S \geq L$, this implies that $R \geq L$ and so $R \in (L,U)$ and $R' \in \pre^*(L,U)$.

Finally, we show that the norms of $(L',U')$ are bounded.
For the $L$-norm, we simply add up the tokens in $S=L'$. 
Thus by the Pruning theorem
\begin{align*}
\lnorm{(L',U')} \leq |L| + |P|^3 \leq \lnorm{(L,U)} + |P|^3.
\end{align*}
By definition of the $U$-norm, 
$\unorm{(L',U')} = \sum_{\substack{p\in P | U'(p)<\infty}} U'(p).$
If $U'(p)<\infty$ then $\mathcal{B}^M_p$ of history $H_M$ is in Case $2$ and there is no bunch $B\in \mathcal{B}^M_p$ going from $p$ to a final place $f_B$ such that $U(f_B)=\infty$.
So the set of bunches $B$ starting in a place $p$ such that $U'(p)<\infty$ is included in the set of bunches $B'$ such that $U(f_{B'})<\infty$, and thus
\begin{align*}
\sum_{\substack{p\in P | U'(p)<\infty}} U'(p) = \sum_{\substack{p\in P | U'(p)<\infty}} \left( \sum_{B \in \mathcal{B}^M_p} size(B) \right)
\leq \sum_{B | U(f_B)<\infty} size(B).
\end{align*}
Now $\sum_{B | U(f_B)<\infty} size(B)$ in history $H_M$ is exactly $\sum_{\substack{p\in P | U(p)<\infty}} M(p)$.
Since $M \in (L,U)$, for all places we have $M(p)\leq U(p)$ and so 
\begin{align*}
\sum_{\substack{p\in P | U'(p)<\infty}} U'(p) \leq \sum_{\substack{p\in P | U(p)<\infty}} M(p)
\leq \sum_{\substack{p\in P | U(p)<\infty}} U(p).
\end{align*}
So by definition of the norm, $\unorm{(L',U')} \leq \unorm{(L,U)}$. \qed
\end{proof}

This result entails the main theorem of the section.

\begin{restatable}{theorem}{ThmCCReachability}
\label{thm:new18}
Let $N$ be an IO net with a set $P$ of places, and let $S$ be a counting set.
Then $pre^*(S)$ is a counting set and there exist counting constraints $\Gamma$ and $\Gamma'$ satisfying $\sem{\Gamma} = S$, $\sem{\Gamma'} = \pre^*(S)$ and we can bound the norm of $\Gamma'$ by
 \begin{align*}
\unorm{\Gamma'} \leq \unorm{\Gamma} \text{ and }
\lnorm{\Gamma'} \leq  \lnorm{\Gamma} + |P|^3
 \end{align*}
The same holds for $\post^*$ by using the net with reversed transitions.
\end{restatable}
\begin{proof}[Sketch]
Lemma \ref{lm:smallminterm} gives ``small" cubes such that $\pre^*(S)$ is the union of these cubes.
Since there are only a finite number of such ``small" cubes, this union is finite and $\pre^*(S)$ is a counting set. 
The bounds on the norms of $\pre^*(S)$ are derived from the bounds on the norms of these cubes.
\end{proof}

\begin{remark}
Theorem \ref{thm:new18} is a dramatic improvement on Theorem 18 of \cite{EsparzaGMW18},
which could only give a much higher bound for $\lnorm{\Gamma'}$: \\
\centerline{$\lnorm{\Gamma'} \leq (\lnorm{\Gamma} + \unorm{\Gamma})^{2^{\mathcal{O}(|P|^2 \log |P|)}}$ instead of $\lnorm{\Gamma'} \leq  \lnorm{\Gamma} + |P|^3$.}
\end{remark}

\section{Cube Problems for IO Nets Are in \PSPACE}
We prove that the cube-reachability, cube-coverability, and cube-liveness problems for IO nets are in \PSPACE.

\begin{theorem}
\label{thm:set-seteasy}
The cube-reachability and cube-coverability problems for IO nets are in \PSPACE.
\end{theorem}
\begin{proof}
Let us first consider cube-reachability.
Let $N$ be an IO net with set of places $P$, and let $S_0$ and $S$ be cubes.
Some marking of $S$ is reachable from some marking of $S_0$ if{}f $\post^*(S_0) \cap S \neq \emptyset$. 
Let $\Gamma_0$ and $\Gamma$ be two counting constraints for $S_0$ and $S$ respectively.
By Theorem~\ref{thm:new18}  and Proposition~\ref{prop:oponconf}, there exists a counting constraint $\Gamma'$ such that $\sem{\Gamma'}=\post^*(S_0) \cap S$, and such that $\unorm{\Gamma'} \leq \unorm{\Gamma_0} + \unorm{\Gamma}$ and $\lnorm{\Gamma'} \leq \lnorm{\Gamma_0} + |P|^3 + \lnorm{\Gamma}$.
Therefore, $\post^*(S_0) \cap S \ne \emptyset$ holds
if{}f \(\post^*(S_0) \cap S\) contains a  “small” marking $M$ satisfying \(|M| \leq  \lnorm{\Gamma_0} +  |P|^3 + \lnorm{\Gamma} \). 
The \PSPACE{} decision procedure takes the following steps:
\textsf{\textbf{1)}}    Guess a “small” marking $M \in S$.
\textsf{\textbf{2)}}  Check that \(M\) belongs to \(\post^*(S_0)\).

The algorithm for \textsf{\textbf{2)}} is to guess a marking \(M_0 \in S_0\) such that \(|M_0|=|M|\), and then guess a firing sequence  (step by step), leading from \(M_0\) to $M$. 
This can be performed in polynomial space because each marking along the path is of size \(|M|\), and we only need to store the current marking to check if it is equal to \(M\).

Now for cube-coverability. 
Again let $N$ be an IO net with set of places $P$, and let $S_0$ and $S$ be cubes.
In particular let $S=(L,U)$ for some upper and lower bounds $L,U$.
Some marking of $S$ is coverable from some marking of $S_0$ if{}f $\post^*(S_0) \cap S_{\infty} \neq \emptyset$, where $S_{\infty}$ is the cube defined by lower bound $L$ and upper bound $\infty$ on all places.
From here we proceed with the same \PSPACE{} decision procedure as above.
\qed
\end{proof}

Notice that cube-reachability and coverability can be extended to counting set-reachability and coverability simply by virtue of a counting set being a finite union of cubes.

Recall that a marking $M_0$ of an IO net $N$ is \emph{live} if for every marking $M$ reachable from $M_0$ and for every transition $t$ of $N$, some marking reachable from $M$ enables $t$.
The cube-liveness problem consists of deciding if, given a net $N$ and a cube ${\cal M}$ of markings of $N$, every marking of ${\cal M}$ is live.

\begin{theorem}
\label{thm:livepspace}
The cube-liveness problem for IO nets is in \PSPACE.
\end{theorem}
\begin{proof}
Let $N$ be an IO net with set of places $P$, and $\mathcal{M}$ a cube.
Let $t=(p_s,p_o)\mapsto (p_d,p_o)$ be a transition of $N$.
The set $En(t)$ of markings that enable $t$ contains the markings that put at least one token in $p_s$ and at least one token in $p_o$ (unless $p_s=p_o$ in which case there should be at least two tokens in that place). Clearly, $En(t)$ is a cube.
Then $\overline{\pre^*(En(t))}$ is the set of markings $M$ from which one cannot execute 
transition $t$ anymore by any firing sequence starting in $M$. So the set $\mathcal{L}$ of live markings of $N$ is given by
$$
\mathcal{L} = \overline{ \pre^*\left( \bigcup_{t \in T} \overline{\pre^*(En(t))} \right) }
$$
Deciding whether $\mathcal{M} \subseteq \mathcal{L}$ is equivalent to deciding whether  $\mathcal{M} \cap \overline{\mathcal{L}} = \emptyset$ holds, or, equivalently, whether
$\bigcup_{t \in T} \overline{\pre^*(En(t))}$ is reachable from $\mathcal{M}$. 
By definition, the cube describing $En(t)$ has an L-norm equal to $2$ and U-norm equal to $0$.
By Theorem~\ref{thm:new18}  and Proposition~\ref{prop:oponconf}, there exists a counting constraint $\Gamma'$ such that $\sem{\Gamma'}=\bigcup_{t \in T} \overline{\pre^*(En(t))}$ and its norms are of size polynomial in $|P|$.
So by Theorem \ref{thm:set-seteasy} this reachability problem can be solved in \PSPACE \ in the size of the input, i.e. net $N$ and set $\mathcal{M}$.
\qed
\end{proof}

\section{Application: Correctness of IO Protocols is \PSPACE-complete}
\label{Application}
In \cite{EsparzaGMW18}, Esparza \emph{et al.} studied the correctness problem for immediate observation protocols. The problem asks, given a protocol and a predicate, whether the protocol computes the predicate.  In order to study the complexity of the problem we need to restrict ourselves to a class of predicates representable by finite means. Fortunately, Angluin \textit{et al.} have shown in \cite{AAER07} that IO protocols compute exactly the predicates representable by counting constraints, i.e., the predicates $\varphi \colon \mathbb{N}^k \rightarrow \{0,1\}$ for which there is a counting constraint $\Gamma$  such that  $\varphi(\vec{v})=1$ if{}f $\vec{v}$ satisfies $\Gamma$. So we can formulate the problem as follows: given a counting constraint $\Gamma$ and an IO protocols with a suitable set of input states, does it compute the predicate described by $\Gamma$? It is shown in  \cite{EsparzaGMW18} that the problem is  \PSPACE-hard and in \EXPSPACE, and closing this gap was left for future research.

In Petri net terms, the correctness problem for IO nets asks, given an IO net $N$ and a counting constraint $\Gamma$, whether $N$ computes $\Gamma$ (formally defined in Section \ref{IOPPprimer}). We use the Pruning Theorem and the results of this paper to show that the correctness problem for IO nets, and so for IO protocols, is \PSPACE-complete.

We present a proposition that characterizes the nets $N$ that compute a given predicate $\varphi \colon \mathbb{N}^k \rightarrow \{0,1\}$. On top of the definitions of Section \ref{IOPPprimer}, we need some notations. For $b \in \{0,1\}$:
\begin{itemize}
\item $\mathcal{I}_{b} = \{\Mv \mid \varphi(\vec{v})=b\}$, i.e., $\mathcal{I}_{1}$ ($\mathcal{I}_{0}$) denotes the initial markings of $N$ for the input vectors satisfying (not satisfying) $\varphi$.
\item $\mathcal{C}_b$ denotes the set of $b$-consensuses of $N$. 
\item $\mathcal{ST}_b \defeq \overline{\pre^*\left(\overline{\mathcal{C}_b}\right)}$ denotes the set of stable consensuses of $N$ (the complement of the markings from which one can reach a non-$b$-consensus).
\end{itemize}

\begin{restatable}{proposition}{PropCharacterization}
\label{prop:correctness}
Let $N$ be an IO net, let $I$ be a set of input places, and let $\varphi: \mathbb{N}^{k} \rightarrow \{0,1\}$ be a predicate where $k=|I|$.
Net $N$ computes $\varphi$ if{}f $\post^*(\mathcal{I}_{b}) \subseteq \pre^*(\mathcal{ST}_b )$ holds for $b \in \{0,1\}$.
\end{restatable}

We can now show:

\begin{restatable}{theorem}{ThmMain}
\label{thm:main}
The correctness problem for IO nets is \PSPACE-complete.
\end{restatable}
\begin{proof}
Let $N$ be an IO net with \(P\) its set of places, $I$ a set of input places of size $k$, and $\varphi \colon \mathbb{N}^k \rightarrow \{0,1\}$ a predicate described by some counting constraint $\Gamma_{\varphi}$.
Recall that \(\mathcal{ST}_b\) is given by \(\overline{ \pre^*( \overline{ \mathcal{C}_b } ) }\) where $\mathcal{C}_b$, for $b \in \{ 0,1 \}$, can be represented by the cube defined by the upper bound equal to $0$ on all places $p_i\in O^{-1}(1-b)$ and $\infty$ otherwise, and the lower bound equal to $0$ everywhere. 
The condition for correctness of Proposition \ref{prop:correctness} can be rewritten as 
\begin{equation}
\label{correct}
\post^*(\mathcal{I}_{b}) \cap \overline{\pre^*(\mathcal{ST}_b )} = \emptyset.
\end{equation}
Deciding (\ref{correct}) is equivalent to deciding whether $\overline{\pre^*(\mathcal{ST}_b )}$ is reachable from $\mathcal{I}_{b}$.
The cube describing $\mathcal{C}_b$ has upper and lower norm equal to $0$.
By Theorem~\ref{thm:new18}  and Proposition~\ref{prop:oponconf}, there exists a counting constraint $\Gamma_b$ such that $\sem{\Gamma_b}=\overline{\pre^*(\mathcal{ST}_b )}$ and its norms are of size polynomial in $|P|$.
Set $\mathcal{I}_{b}$ is a counting set described by either $\Gamma_{\varphi}$ or its complement.
So by Theorem \ref{thm:set-seteasy} this reachability problem can be solved in \PSPACE.

The proof for \PSPACE-hardness reduces from the acceptance problem for 
deterministic Turing machines running in linear space, and is in the full version \cite{EsparzaRW19}.
\qed
\end{proof}

\section{Conclusion}
Many modern distributed systems are parameterized, and they have to be modeled as an infinite set of Petri nets differing only in their initial markings. This leads to a new class of \emph{parameterized} analysis problems, which typically are much harder to solve that standard ones.  We have shown that, remarkably, this is not the case for immediate observation Petri nets, a subclass of $\vec{1}$-conservative nets  able to model immediate observation protocols and enzymatic chemical reaction networks. We have proved that the parameterized reachability, coverability, and liveness problems are \PSPACE-complete, which is also the complexity of their non-parameterized versions. Current research on population protocols or networks considers quantitative properties like, in the case of population protocols, the computation of the expected time to stabilization. In future research we plan to study algorithms for these questions.

\medskip

\noindent \textbf{Acknowledgments.} We thank three anonymous reviewers for numerous suggestions to improve readability, and Pierre Ganty for many helpful discussions.

%

\appendix
\section{Appendix for Section \ref{Parameterized}}
\ThmConservativeNets*
\begin{proof}
The first part is proved (modulo straightforward modifications) in \cite{ChengEP95,Esparza96}. For the second part, let $N=(P,T,F)$ be an arbitrary Petri net. We construct a Petri net $N'=(P \cup \{r,s\}, T, F')$, where $r$ and $s$ are two new places, the \emph{repository} and the \emph{sink}, and $F'$ is defined so that, intuitively,
transitions of $N'$ neither create nor destroy tokens. Formally, for every transition $t$:
\begin{itemize}
\item $F'(p, t)=F(p,t)$ and $F'(t, p)=F(t,p)$ for every $p \in P$.
\item $F'(t, r)=0$, and $F'(r, t) = \max\{ 0,  |\postset{t}|-|\preset{t}|\}$.
\item $F'(s, t)=0$, and $F'(t, s) = \max\{ 0,  |\preset{t}|-|\postset{t}|\}$.
\end{itemize}
In $N'$ we have $\sum_{p \in P \cup \{r,s\}} F'(p, t) = \sum_{p \in P \cup \{r,s\}} F'(t, p)$ for every transition $t$, and so $N'$ is conservative.

Given a marking $M$ of $N$, let $(L_M, U_M)$ be the cube of $N'$ given by $L_M(p) = M(p) = U_M(p)$ for every $p \in P$, $L_M(r)=L_M(s)=0$ and $U_M(r)=U_M(s)=\infty$. Clearly, we have: $M_2$ is reachable (coverable) from $M_1$ in $N$ if{}f $(L_{M_2},U_{M_2})$ is reachable (coverable) from $(L_{M_1},U_{M_1})$ in $N'$, and we are done.
\qed
\end{proof}

\section{Appendix for Section \ref{IO-nets-hardness}}

To remind the notation, let us start with an illustration of transitions modelling a single step.

\begin{figure}
\centering
\includegraphics[width=10cm]{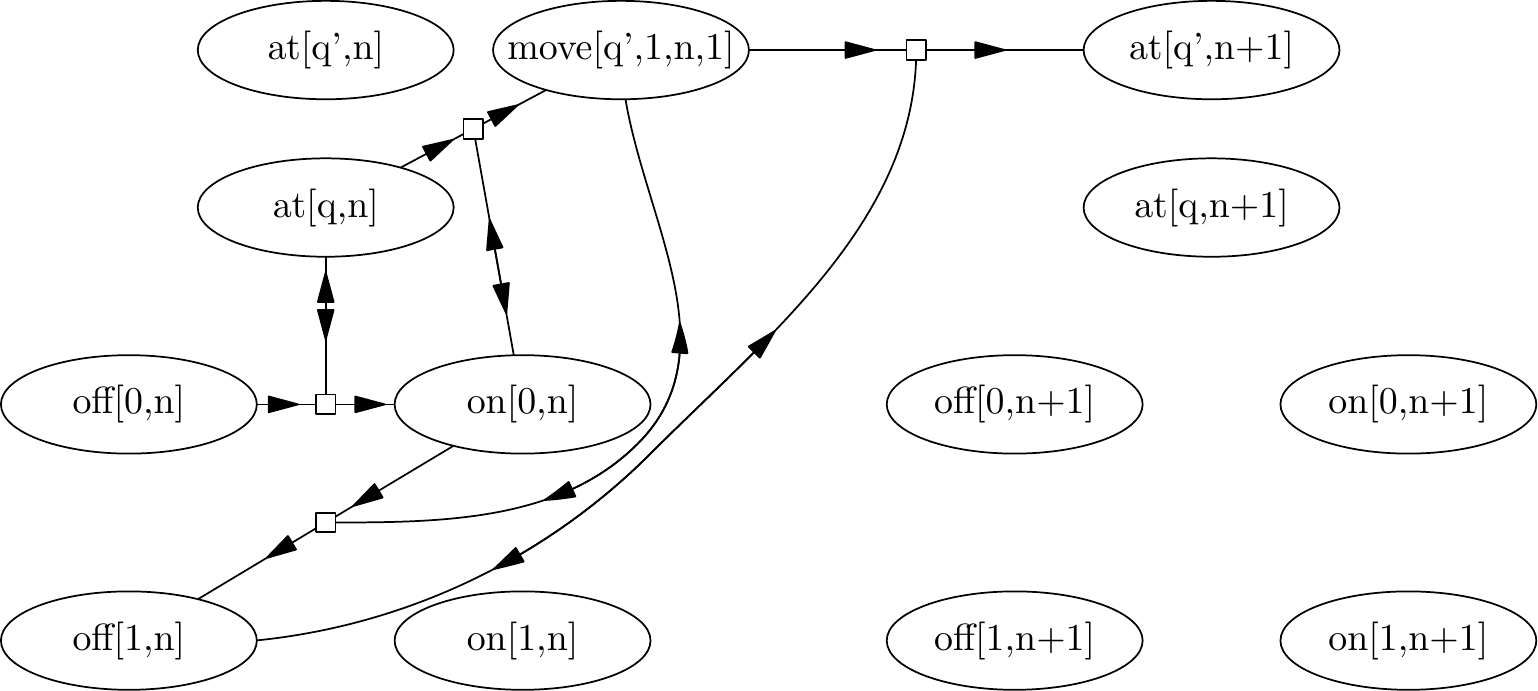}
\caption{Some of the places and transitions involved in modelling a Turing machine}
\label{figure-turing-io}
\end{figure}

Figure~\ref{figure-turing-io} illustrates transitions involved in modelling a single step of a Turing machine
that reads $0$, writes $1$, moves head to the right and switches the control state from $q$ to $q'$.

\begin{definition}
 A marking of $N_M$ is a \emph{modelling marking} if the following conditions hold.
\begin{enumerate}
        \item For every $1\leq n\leq K$ exactly one of the $2|\Sigma|$ places 
        $\act{\sigma}{n}, \pass{\sigma}{n}$ is marked, and marked with a single token. \\
        (Intuitively: every cell is either \textit{on} or \textit{off} and contains exactly one symbol.)
        \item Exactly one of all the head places is marked (again, with a single token).
        \item If a cell place $\act{\sigma}{n}$ is marked, then a head place $\stable{q}{n}$ or $\switch{q}{\sigma'}{n}{d}$ is marked for some $\sigma'$ and $d$.
        \item If a head place $\switch{q}{\sigma}{n}{d}$ is marked,
                either $\act{\sigma'}{n}$ is marked for some $\sigma'$,
                of $\pass{\sigma}{n}$ is marked.
\end{enumerate}
\end{definition}

\begin{remark}
Note that for every configuration $c$ of $M$ the marking $M_c$ is a modelling marking.
\end{remark}

\begin{lemma}
        \label{lemma:modellingmarkingevolution}
For every modelling marking $M$ of $N_M$:
\begin{itemize}
\item[(1)] $M$ enables at most one transition.
\item[(2)] If $M$ enables no transitions, then it marks places $\act{\sigma}{n}$ and $\stable{q}{n}$ for some $q \in Q$, $\sigma \in \Sigma$, and $1 \leq n \leq K$.
\item[(3)] If $M \trans{} M'$, then $M'$ is also a modelling marking. 
\end{itemize}
\end{lemma}
\begin{proof}
\noindent (1) All possible transitions require tokens at two places,
        one of type $\act{\cdot}{n}$ or $\pass{\cdot}{n}$
        and one of type $\stable{\cdot}{n}$ or $\switch{\cdot}{n}{\cdot}{\cdot}$, with the same $n$.
        But the modelling condition requires that there can be at most one such pair.

\noindent (2) If a $\switch{\cdot}{\cdot}{\cdot}{\cdot}$ place is marked,
        a transition is always possible by definition
        of the list of $\switch{\cdot}{\cdot}{\cdot}{\cdot}$ places.
        The same for the case where a $\stable{\cdot}{\cdot}$ place
        is marked but no $\act{\cdot}{\cdot}$ case is marked.
        If there are marked places of types $\act{\cdot}{n}$ and $\stable{\cdot}{n}$,
        the transition may fail
        to exist if either the Turing machine halts or if it goes outside the allocated space.

\noindent (3) Every transition consumes and produces one token 
        at $\pass{\cdot}{n}$ or $\act{\cdot}{n}$ place,
        and the new place has the same $n$.
        Every transition consumes and produces one token
        at $\switch{\cdot}{\cdot}{\cdot}{\cdot}$ or $\stable{\cdot}{\cdot}$ place.
        If an $\act{\cdot}{n}$ place becomes marked after a transition,
        it has the same $n$ as the marked $\stable{\cdot}{n}$ place
        of both markings (before and after);
        if an $\act{\cdot}{n}$ place stays marked,
        the token is moved from a $\stable{\cdot}{n}$ to a $\switch{\cdot}{n}{\cdot}{\cdot}$
        place with the same $n$.
When $\switch{q}{\sigma}{n}{d}$ becomes marked,
        the transition needs a marked $\act{\cdot}{n}$ place.
        When $\switch{q}{\sigma}{n}{d}$ stays marked,
        the transition marks a $\pass{\sigma}{n}$ place.
        \qed
\end{proof}

\theoremSimulationStep*
\begin{proof}
        By Lemma~\ref{lemma:modellingmarkingevolution},
        for all $c$ there is either zero or one possibility for the
sequence $t_1,t_2,t_3,t_4$ starting in $M_c$. 
        It is easy to see from the definition of steps marking $\switch{\cdot}{\cdot}{\cdot}{\cdot}$
        places that if such a sequence exists, it results in $c'$ such that $c\trans{}c'$.
        If such a sequence doesn't exist, the failure must occur when trying to mark
        a  $\switch{\cdot}{\cdot}{\cdot}{\cdot}$ place. 
        In that case the configuration~$c$ must be blocked, either by
        the transition being undefined or by going out of bounds.
\end{proof}

\ThmReachabilityHard*
\begin{proof}
The proof is routine. Let $p$ be a fixed polynomial satisfying $p(n) \geq n$ for all $n$. 
Consider the set of deterministic Turing machines whose set of states contains two distinct distinguished states $q_{acc}, q_{rej}$,  and whose computation on empty tape satisfies the following conditions:
\begin{itemize}
\item The computation never visits a configuration that visits more than $p(n)$ cells, where $n$ is the size of $M$,
and visits the set $\{q_{acc}, q_{rej}\}$ of states exactly once.
\item The computation ends in a configuration $c$ with empty tape, head on the first cell, and control state either $q_{acc}$ or $q_{rej}$. 
\end{itemize}
We say that the machine \emph{accepts} (\emph{rejects}) if it terminates in $q_{acc}$ ($q_{rej}$). 
It is well known that the problem whether such a machine accepts on empty tape is \PSPACE-hard. Given such a machine $M$, let $N_M$ be its associated IO net, and let $M_0$ and $M$ be the modeling markings describing the initial configuration and the unique accepting configuration.  Then $M$ accepts if{}f $M$ is reachable from $M_0$ if{}f some marking reachable from $M_0$ covers the marking that puts a token in the place for $q_{acc}$.

Now we reduce termination of bounded-tape Turing machines to liveness of 
immediate-observation Petri nets.
        Consider a Turing machine $M$ with the accepting state $q_{acc}$. 
First, we convert it to an immediate-observation
Petri net as before. Afterwards, we add two additional places, $observer$
        and $success$. We add the following transitions: 
        \begin{itemize}
	\item$(\stable{q_{acc}}{n},observer) \mapsto (\stable{q_{acc}}{n},success$), and
	\item$(success, *) \mapsto (success, *)$.
	\end{itemize}
        Initially, we place the tokens according to the initial control state
        and tape contents, and additionally put one token into observer.
        Now, if the Turing machine cannot reach the accepting state, the net
        will never be able to execute any transition into $success$
        (so it will not be live).
        If the Turing machine can reach the accepting
        state, the only possible sequence of transitions of the net will
        lead to marking of some place $\stable{q_{acc}}{n}$. Afterwards, 
        the net can optionally switch a tape state from passive to active, 
        but cannot continue further without activating a transition that 
        marks the $success$ place. 
        
        As our Petri net contains at least two
        other tokens, and as $success$ place is such a trap that marking 
        it allows moving tokens between any two places, firing
        this transition makes it possible to mark two arbitrary places
        from any later marking, which allows to fire any transition.
        Therefore if the Turing machine reaches the accepting state, the Petri net 
        is live.

        We have proven the reduction of the acceptance problem 
        for Turing machines running in linear space to 
        liveness of immediate-observation Petri nets, which implies
        \PSPACE-hardness.
        \qed
\end{proof}

\section{Appendix for Section \ref{Pruning}}

\LemmaRealizable*
\begin{proof}
One direction is obvious by definition: if we have a realizable history (even not well-structured),
it also describes a firing sequence.
Let us prove the other direction.

Informally, we just implement the de-anonymisation.
A formal proof can be given by induction in the number of transitions in the firing sequence.

\noindent
\emph{Base case}. If there are no transitions, we create a multiset of trajectories of length one such that the initial places of the trajectories are exactly the places (with multiplicity) of the initial marking of the firing sequence.
This is well-structured because there are no steps.

\noindent
\emph{Induction step}. Consider a sequence of transitions and a corresponding well-structured $N$-history.
Now let us add a single enabled transition. 
To build the new history, we choose an arbitrary trajectory of the existing history such that this trajectory ends in the place corresponding to the source place of the added transition.
Such a trajectory exists because the transition is enabled and therefore its source place must be marked.
We extend the chosen trajectory with a step from the source place to the destination place of the added transition,
and we extend the rest of the trajectories with one horizontal step each.
We obtain a multiset of trajectories of same length, thus constituting a history.
It is realizable using the considered sequence of transitions followed by the new enabled transition.
As we add only a single non-horizontal step at that moment of time, we cannot
break the well-structuring condition.
\qed
\end{proof}

\LemmaHistory*
\begin{proof} 
Let a well-structured history $H$ be realizable in $N$.
Consider an arbitrary non-horizontal step $\tau(i) \tau(i+1)=p_s p_d$ in some trajectory of this history.
All the non-horizontal steps at the corresponding position in $H$ are equal by well-structuredness,
and realizability implies that there is an enabled transition
with source place $p_s$ and destination place $p_d$ at marking $M_H^i$ in $N$.
This transition can be applied as many times as there are equal steps 
at the corresponding position in $H$.
Therefore the observed place $p_o$ of this transition is marked both before and after
iterating this transition,
which corresponds to $H$ containing a trajectory with the step $p_o p_o$
at the corresponding position.
As this holds for each non-horizontal step in $H$, $H$ is compatible with $N$.

Now assume that $H$ is compatible with $N$.
If some position in $H$ contains only horizontal steps,
we can use zero iterations of an arbitrary transition.
If a position contains some number of (equal) non-horizontal steps
$p_s p_d$, it also contains a horizontal step $p_o p_o$
such that $(p_s,p_o)\mapsto(p_d,p_o)$ is a transition in $N$.
All the other steps at the corresponding position are horizontal.
Therefore we can iterate the transition $(p_s,p_o)\mapsto(p_d,p_o)$ 
to obtain the next marking.
\qed
\end{proof}

\begin{theorem} [Quadratic Pruning Theorem]
Let $N = (P,T,F)$ be an IO net, let $M$ be a marking of $N$, and let
$M'' \trans{*} M'$ be a firing sequence of $N$ such that $M' \geq M$.
There exist markings $S''$ and $S'$ such that 
\begin{center}
\(
\begin{array}[b]{@{}c@{}c@{}c@{}c@{}c@{}c@{}c@{}}
M'' &  \trans{\hspace{1em}*\hspace{1em}} & M' & \  \geq \ &M  \\[0.1cm]
\geq &  & \geq \\[0.1cm]
S'' & \trans{\hspace{1em}*\hspace{1em}} &S'  & \geq  & M
\end{array}
\)
\end{center}
and $|S''| \leq |M| + 2|P|^2$.
\end{theorem}

\begin{proof}
The proof is similar to the proofs of Lemma~\ref{lm:pruning} and Theorem~\ref{thm:pruning}.
The main difference is the following.
In Lemma~\ref{lm:pruning} we keep trajectories that belong to small bunches,
and prune each large bunch separately.
To prove the quadratic lower bound we keep trajectories from and to small places,
then prune all the remaining trajectories together.
The place is called small if it has less than $|P|$ incoming or outgoing trajectories.

Let $M'' \trans{*}  M'  \  \geq  \ M$. 
By Lemma \ref{lem:realizable}, there is a well-structured realizable history $H$ with $M''$ and $M'$ as initial and final markings, respectively. Let $H_M \subset H$ be an arbitrary sub(multi)set of $H$ with final
marking $M$, and initially set $H'=H-H_M$.
We further reduce $H'$ by repeatedly removing all the trajectories with initial or final place
having less than $|P|$  trajectories still in $H'$. 
We can perform at most $2|P|$ steps like that, removing at most $|P|-1$ trajectories per step.
At the end, we will add back these trajectories as well as those of $H_M$.

Now we can define $Q$ as the set of all places reached by the remaining trajectories in $H'$,
and $f(q)$ and $l(q)$ for $q\in Q$ 
as the earliest and the latest moment in time when this place has been used by any of the trajectories
(possibly on different trajectories, and possibly on trajectories with different initial and final place).

We now build a trajectory for every $q\in Q$ by reaching it by the moment $f(q)$ and leaving it after $l(q)$.
As all the trajectories in $H'$ have initial and final place with at least $|P|$ trajectories in $H'$,
the set of trajectories that we build will have the initial and final markings covered
by the corresponding markings of $H'$.

The rest of the proof is identical to the proofs
of Lemma~\ref{lm:pruning} and Theorem~\ref{thm:pruning}.
\qed
\end{proof}

\section{Appendix for Section \ref{CC}}

\ThmCCReachability*
\begin{proof}
Lemma \ref{lm:smallminterm} states that for every cube $C$ of a finite decomposition into cubes of $S$, for every marking $m$ in $\pre^*(C)$, there is a ``small" cube $C'_m$ such that $m$ is in $C'_m$ and $C'_m$ is completely in $\pre^*(C)$.
So $\pre^*(C) = \cup_{m \in \pre^*(C)} C'_m$.
But there are only a finite number of such "small" cubes.

So $\pre^*(C)$ is a finite union of cubes.
There exists some finite $k$ such that $\pre^*(S) = \cup_{i=1}^k \pre^*(C_i)$. 
Each of these $\pre^*(C_i)$ is itself a finite union of cubes, so $\pre^*(S)$ is a finite union of cubes.
Thus by definition, $\pre^*(S)$ is a counting set.

Let $\Gamma$ be the counting constraint defined as the union of the $C_i$.
Let $\Gamma'$ be the counting constraint defined as the union of the $\pre^*(C_i)$, themselves unions of ``small" cubes.
Then by the bounds in Lemma \ref{lm:smallminterm} and by definition of the norms,
$
\unorm{\Gamma'} \leq \unorm{\Gamma}
$ 
and 
$
\lnorm{\Gamma'} \leq \lnorm{\Gamma} + |P|^3.
$

The results also hold for $\post^*(S)$. 
Consider the IO net $N_r$, the ``reverse" of net $N = (P,\Delta)$. 
Net $N_r$ is defined as $N$ but with transition set $\Delta_r$, where $\Delta_r$ has a transition $(p_1, p_2) \mapsto (p_3, p_4)$ if{}f $\Delta$ has a transition $(p_3, p_4) \mapsto (p_1, p_2)$. 
Notice that $N_r$ is still an IO net.
Then $\post^*(S)$ in $N$ is equal to $\pre^*(S)$ in $N_r$.
\qed
\end{proof}

\section{Appendix for Section \ref{Application}}

\PropCharacterization*
\begin{proof}
$N$ computes $\varphi$ if for $b=0,1$, for every initial marking $\Mv$ such that $\varphi(\vec{v})=b$ (i.e. $\Mv \in \mathcal{I}_{b}$), every fair firing sequence starting in $\Mv$ converges to $b$. 
We call $(Def)$ this condition.
Let us call $(A)$ the following condition: for every $M \in \post^*(\mathcal{I}_{b})$, there exists $M' \in \mathcal{ST}_b $ such that $M'$ is reachable from $M$.
Let us show that $(A)$ is equivalent to $(Def)$.
Assume we have $(Def)$, and let $M \in \post^*(\mathcal{I}_{b})$.
Then there exists $M_0$ in $\mathcal{I}_{b}$ such that $M_0 \trans{*} M$. 
Extend it into a fair firing sequence. 
By $(Def)$, the firing sequence converges to $b$, so $\mathcal{ST}_b$ is reachable from every marking of the firing sequence.
We reuse a lemma from \cite{EsparzaGMW18} (Lemma 21) which states that given an infinite fair firing sequence $M_0,M_1,M_2 \ldots$ of an IO net and a set $S$ of markings, if $S$ is reachable for infinitely many indices $j\geq0$ then $M_j\in S$ for infinitely many $j \geq 0$.
We apply this lemma to our fair firing sequence: since $\mathcal{ST}_b$ is reachable from every marking, the firing sequence reaches a marking of $\mathcal{ST}_b$.
Now assume we have $(A)$, let us show it implies $(Def)$.
Consider a fair firing sequence starting in $M_0 \in \mathcal{I}_{b}$.
By $(A)$ and Lemma 21 of \cite{EsparzaGMW18}, the firing sequence reaches a marking in $\mathcal{ST}_b$.
From $\mathcal{ST}_b$ one can only reach other markings of $\mathcal{ST}_b$ and so the firing sequence converges to $b$.
So $(A)$ is equivalent to $(Def)$, and $(A)$ can be written $\post^*(\mathcal{I}_{b}) \subseteq \pre^*(\mathcal{ST}_b )$.
\qed
\end{proof}

\ThmMain*
\begin{proof}
The proof that the correctness problem is in \PSPACE \ is in the main part of the paper. Here we prove that the correctness problem is \PSPACE-hard, using the construction from the proof of Theorem \ref{thm:liveness-hard}.
        
        Given a Turing machine with initial state $q_{init}$ and a size bound,
        we construct a corresponding IO net and apply some changes.
        We restrict $(success,*)$ transitions to $(success,*) \mapsto (success,success)$.
        We also add transitions such that if there are two tokens in the head places,
        or two tokens in the cell places for the same cell,
        or two tokens in the $observer$ place,
        one of them can move to $success$.
        The input places are $\pass{0}{\cdot}$, $\stable{q_{init}}{1}$ and the $observer$ place.
        The output function is $1$ for $success$ and $0$ otherwise.
        We define a predicate as ``there are at least two tokens in the head places,
        or at least two tokens in the cell places for some cell''.

        If the Turing machine accepts the empty tape without going out of bounds,
        the protocol is not correct, as we can put exactly one token in every 
        input and run the simulation until the acceptance will lead to one of the
        $(\stable{q_{acc}}{\cdot},observer) \mapsto (\stable{q_{acc}}{\cdot},success)$ transitions firing.

        Otherwise the protocol is correct, as there are markings not greater than the
        marking with one token in every input place, which cannot mark the $success$ 
        state because the bounding marking cannot;
        the remaining markings are accepted by the predicate and will also converge
        to all the tokens being in the $success$ place.
        \qed
\end{proof}

\end{document}